\documentclass[11pt, a4paper,onecolumn]{belarticle4}
\pdfoutput=1
\usepackage[utf8]{inputenc}
\usepackage[english]{babel}
\usepackage[T1]{fontenc}
\usepackage{amsmath}
\usepackage{xcolor}
\colorlet{myPurple}{blue!40!red}
\colorlet{myCyan}{cyan!60!gray}
\colorlet{myRed}{blue!55!gray}
\usepackage{tikz}
\usepackage{pgfplots}
\pgfplotsset{compat=1.14}
\usepackage[colorlinks=true,citecolor=myRed,urlcolor=myRed,linkcolor=myRed]{hyperref}
\usepackage{exscale}
\usepackage{bbm}
\usepackage{graphicx}
\usepackage{amsmath}
\usepackage{latexsym}
\usepackage{amsfonts}
\usepackage{amssymb}
\usepackage[T1]{fontenc}
\usepackage{lipsum}
\usepackage{amsthm}
\usepackage{enumerate}
\usepackage{bbold}
\usepackage{color}
\usepackage{nicefrac}%braket stuff
\usepackage{changes}
\usepackage{bm}
\usepackage{dsfont}
\usepackage{amsthm}
\usepackage{array}
\usepackage{cancel}
\usepackage[toc,page]{appendix}
\usepackage{multirow}
\usepackage{color}
\usepackage{calrsfs}
\usetikzlibrary{backgrounds,decorations.pathreplacing,calc}

\usepackage{tkz-euclide}
\usepackage{tcolorbox}
\usepackage{mathtools}
\usepackage{paralist}
\usepackage{cite}

\newcommand{\sket}[1]{{\ensuremath{\lvert#1\rangle}}}
\newcommand{\lket}[1]{{\ensuremath{\left\lvert#1\right\rangle}}}
\newcommand{\ket}[1]{\if@display\lket{#1}\else\sket{#1}\fi}

\newcommand{\sbra}[1]{{\ensuremath{\langle#1\rvert}}}
\newcommand{\lbra}[1]{{\ensuremath{\left\langle#1\right\rvert}}}
\newcommand{\bra}[1]{\if@display\lbra{#1}\else\sbra{#1}\fi}

\newcommand{\sbraket}[2]{{\ensuremath{\langle#1\rvert#2\rangle}}}
\newcommand{\lbraket}[2]{{\ensuremath{\left\langle#1\!\left\rvert\vphantom{#1}#2\right.\!\right\rangle}}}
\newcommand{\braket}[2]{\if@display\lbraket{#1}{#2}\else\sbraket{#1}{#2}\fi}

\newcommand{\tr}{\textrm{tr}}

\newcommand{\cH}{\mathcal{H}}

\theoremstyle{plain}
\newtheorem{thm}{Theorem}%[section]

\newcommand{\Tr}[2]{\mathrm{tr}_{#2} \left\{ #1 \right\}}
\newcommand{\id}{\mathbb{I}}

\definecolor{quantumviolet}{RGB}{79, 4, 134}

\usepackage{rotating}
\usepackage{floatpag}
\rotfloatpagestyle{empty}
\renewcommand{\theenumi}{\arabic{enumi}}
\usepackage{xcolor}
\definecolor{googleblue}{RGB}{34, 0, 204}

\usepackage{dsfont}
\usepackage{amsmath,amsthm,amssymb}
\usepackage{xspace,enumerate,epsfig}%color,
\usepackage{graphicx}
\graphicspath{{.}{./figures/}}
\usepackage{marvosym}

\usepackage{tikzfig}
\usepackage{stmaryrd}
\usepackage{docmute}
\usepackage{keycommand}

\newkeycommand{\pointmap}[style=point][1]{\,\tikz{\node[style=\commandkey{style}] (x) at (0,-0.3) {$#1$};\node [style=none] (3) at (0, 1.0) {};\node [style=none] (3) at (0, -1.0) {};\draw (x)--(0,0.7);}\,}

\newkeycommand{\typedkpoint}[style=kpoint,edgestyle=][2]{%
\,\begin{tikzpicture}
  \begin{pgfonlayer}{nodelayer}
    \node [style=none] (0) at (0, 1) {};
    \node [style=\commandkey{style}] (1) at (0, -0.75) {$#2$};
    \node [style=right label] (2) at (0.25, 0.5) {$#1$};
  \end{pgfonlayer}
  \begin{pgfonlayer}{edgelayer}
    \draw [\commandkey{edgestyle}] (1) to (0.center);
  \end{pgfonlayer}
\end{tikzpicture}\,}

\newkeycommand{\typedkpointdag}[style=kpoint adjoint,edgestyle=][2]{%
\,\begin{tikzpicture}
  \begin{pgfonlayer}{nodelayer}
    \node [style=none] (0) at (0, -1) {};
    \node [style=\commandkey{style}] (1) at (0, 0.75) {$#2$};
    \node [style=right label] (2) at (0.25, -0.5) {$#1$};
  \end{pgfonlayer}
  \begin{pgfonlayer}{edgelayer}
    \draw [\commandkey{edgestyle}] (0.center) to (1);
  \end{pgfonlayer}
\end{tikzpicture}\,}

\newcommand{\bistateadj}[1]{\,\begin{tikzpicture}[yshift=-3mm]
  \begin{pgfonlayer}{nodelayer}
    \node [style=kpointadj, minimum width=1 cm, inner sep=2pt] (0) at (0, 0.5) {$\psi$};
    \node [style=none] (1) at (-0.75, 0.25) {};
    \node [style=none] (2) at (0.75, 0.25) {};
    \node [style=none] (3) at (-0.75, -0.5) {};
    \node [style=none] (4) at (0.75, -0.5) {};
  \end{pgfonlayer}
  \begin{pgfonlayer}{edgelayer}
    \draw (1.center) to (3.center);
    \draw (2.center) to (4.center);
  \end{pgfonlayer}
\end{tikzpicture}\,}

\newkeycommand{\pointketbra}[style1=copoint,style2=point][2]{\,%
\begin{tikzpicture}
  \begin{pgfonlayer}{nodelayer}
    \node [style=none] (0) at (0, -1.5) {};
    \node [style=\commandkey{style1}] (1) at (0, -0.75) {$#1$};
    \node [style=\commandkey{style2}] (2) at (0, 0.75) {$#2$};
    \node [style=none] (3) at (0, 1.5) {};
  \end{pgfonlayer}
  \begin{pgfonlayer}{edgelayer}
    \draw (0.center) to (1);
    \draw (2) to (3.center);
  \end{pgfonlayer}
\end{tikzpicture}\,}

\newkeycommand{\twocopointketbra}[style1=copoint,style2=copoint,style3=point][3]{\,%
\begin{tikzpicture}
  \begin{pgfonlayer}{nodelayer}
    \node [style=none] (0) at (-0.75, -1.5) {};
    \node [style=none] (1) at (0.75, -1.5) {};
    \node [style=\commandkey{style1}] (2) at (-0.75, -0.75) {$#1$};
    \node [style=\commandkey{style2}] (3) at (0.75, -0.75) {$#2$};
    \node [style=\commandkey{style3}] (4) at (0, 0.75) {$#3$};
    \node [style=none] (5) at (0, 1.5) {};
  \end{pgfonlayer}
  \begin{pgfonlayer}{edgelayer}
    \draw (0.center) to (2);
    \draw (1.center) to (3);
    \draw (4) to (5.center);
  \end{pgfonlayer}
\end{tikzpicture}\,}

\newkeycommand{\twopointketbra}[style1=copoint,style2=point,style3=point][3]{\,%
\begin{tikzpicture}
  \begin{pgfonlayer}{nodelayer}
    \node [style=none] (0) at (0, -1.5) {};
    \node [style=none] (1) at (-0.75, 1.5) {};
    \node [style=\commandkey{style1}] (2) at (0, -0.75) {$#1$};
    \node [style=\commandkey{style2}] (3) at (-0.75, 0.75) {$#2$};
    \node [style=\commandkey{style3}] (4) at (0.75, 0.75) {$#3$};
    \node [style=none] (5) at (0.75, 1.5) {};
  \end{pgfonlayer}
  \begin{pgfonlayer}{edgelayer}
    \draw (0.center) to (2);
    \draw (1.center) to (3);
    \draw (4) to (5.center);
  \end{pgfonlayer}
\end{tikzpicture}\,}

\newkeycommand{\pointbraket}[style1=point,style2=copoint][2]{\,%
\begin{tikzpicture}
  \begin{pgfonlayer}{nodelayer}
    \node [style=\commandkey{style1}] (0) at (0, -0.5) {$#1$};
    \node [style=\commandkey{style2}] (1) at (0, 0.5) {$#2$};
  \end{pgfonlayer}
  \begin{pgfonlayer}{edgelayer}
    \draw (0) to (1);
  \end{pgfonlayer}
\end{tikzpicture}\,}

\newkeycommand{\kpointbraket}[style1=kpoint,style2=kpoint adjoint][2]{\,%
\begin{tikzpicture}
  \begin{pgfonlayer}{nodelayer}
    \node [style=\commandkey{style1}] (0) at (0, -0.75) {$#1$};
    \node [style=\commandkey{style2}] (1) at (0, 0.75) {$#2$};
  \end{pgfonlayer}
  \begin{pgfonlayer}{edgelayer}
    \draw (0) to (1);
  \end{pgfonlayer}
\end{tikzpicture}\,}

\newkeycommand{\kpointmap}[style1=kpoint,style2=map][2]{\,%
\begin{tikzpicture}
  \begin{pgfonlayer}{nodelayer}
    \node [style=\commandkey{style1}] (0) at (0, -1.1) {$#1$};
    \node [style=\commandkey{style2}] (1) at (0, 0.2) {$#2$};
    \node [style=none] (2) at (0, 1.5) {};
  \end{pgfonlayer}
  \begin{pgfonlayer}{edgelayer}
    \draw (0) to (1);
    \draw (1) to (2.center);
  \end{pgfonlayer}
\end{tikzpicture}%
\,}

\newkeycommand{\sandwichtwo}[style1=point,style2=map,style3=map,style4=copoint][4]{\,%
\begin{tikzpicture}
  \begin{pgfonlayer}{nodelayer}
    \node [style=\commandkey{style2}] (0) at (0, -0.75) {$#2$};
    \node [style=\commandkey{style3}] (1) at (0, 0.75) {$#3$};
    \node [style=\commandkey{style1}] (2) at (0, -2) {$#1$};
    \node [style=\commandkey{style4}] (3) at (0, 2) {$#4$};
  \end{pgfonlayer}
  \begin{pgfonlayer}{edgelayer}
    \draw (2) to (0);
    \draw (0) to (1);
    \draw (1) to (3);
  \end{pgfonlayer}
\end{tikzpicture}\,}

\newkeycommand{\longbraket}[style1=point,style2=copoint][2]{\,%
\begin{tikzpicture}
  \begin{pgfonlayer}{nodelayer}
    \node [style=\commandkey{style1}] (0) at (0, -2) {$#1$};
    \node [style=\commandkey{style2}] (1) at (0, 2) {$#2$};
  \end{pgfonlayer}
  \begin{pgfonlayer}{edgelayer}
    \draw (0) to (1);
  \end{pgfonlayer}
\end{tikzpicture}\,}

\newkeycommand{\onb}[style=point]{\ensuremath{\left\{\tikz{\node[style=\commandkey{style}] (x) at (0,-0.3) {$j$};\draw (x)--(0,0.7);}\right\}}\xspace}

\newkeycommand{\copointmap}[style=copoint][1]{\,\tikz{\node[style=\commandkey{style}] (x) at (0,0.3) {$#1$};\draw (x)--(0,-0.7);}\,}

\tikzstyle{cdiag}=[matrix of math nodes, row sep=3em, column sep=3em, text height=1.5ex, text depth=0.25ex,inner sep=0.5em]
\tikzstyle{arrow above}=[transform canvas={yshift=0.5ex}]
\tikzstyle{arrow below}=[transform canvas={yshift=-0.5ex}]

\usetikzlibrary{shapes.multipart}

\tikzstyle{infpoint}=[regular polygon,regular polygon sides=3,draw,scale=0.75,inner sep=-0.5pt,minimum width=9mm,fill=white,regular polygon rotate=90]
\tikzstyle{infcopoint}=[regular polygon,regular polygon sides=3,draw,scale=0.75,inner sep=-0.5pt,minimum width=9mm,fill=white,regular polygon rotate=270]
\tikzstyle{fMeas}=[regular polygon,regular polygon sides=3,draw,scale=0.75,inner sep=-0.5pt,minimum width=4mm,fill=white,regular polygon rotate=0,line width=1pt]
\tikzstyle{fPrep}=[regular polygon,regular polygon sides=3,draw,scale=0.75,inner sep=-0.5pt,minimum width=4mm,fill=white,regular polygon rotate=180,line width=1pt]
\tikzstyle{fTransition}=[fill=white,draw, line width = 1pt,inner sep=0.6mm,font=\footnotesize,minimum height=2mm,minimum width=2mm]
 \tikzstyle{rightground}=[circuit ee IEC,thick,ground,rotate=0,xscale=2.5,yscale=2]

\tikzstyle{bwSpider}=[
       rectangle split,
       rectangle split parts=2,
       rectangle split part fill={black,white},
 minimum size=3.6 mm, inner sep=-2mm, draw=black,scale=0.5,rounded corners=0.8 mm
       ]
 \tikzstyle{wbSpider}=[
       rectangle split,
       rectangle split parts=2,
       rectangle split part fill={white,black},
 minimum size=3.6 mm, inner sep=-2mm, draw=black,scale=0.5,rounded corners=0.8 mm
       ]
\tikzstyle{qWire}=[line width = 1pt, color=quantumviolet]
\tikzstyle{pWire}=[line width = 1pt, color=red!60!black]
\tikzstyle{gWire}=[line width = 1.5pt, color=green!60!black]
\tikzstyle{cWire}=[color=quantumgray,thin]%[densely dotted, thick]
\tikzstyle{env}=[copoint,regular polygon rotate=0,minimum width=0.2cm, fill=black]

\tikzstyle{probs}=[shape=semicircle,fill=white,draw=black,shape border rotate=180,minimum width=1.2cm]

\tikzstyle{every picture}=[baseline=-0.25em,scale=0.5]
\tikzstyle{dotpic}=[] % for backwards-compatibility
\tikzstyle{diredges}=[every to/.style={diredge}]
\tikzstyle{math matrix}=[matrix of math nodes,left delimiter=(,right delimiter=),inner sep=2pt,column sep=1em,row sep=0.5em,nodes={inner sep=0pt},text height=1.5ex, text depth=0.25ex]

\tikzstyle{inline text}=[text height=1.5ex, text depth=0.25ex,yshift=0.5mm]
\tikzstyle{label}=[font=\footnotesize,text height=1.5ex, text depth=0.25ex,yshift=0.5mm]
\tikzstyle{left label}=[label,anchor=east,xshift=1.5mm]
\tikzstyle{right label}=[label,anchor=west,xshift=-1mm]
\tikzstyle{up label}=[label,anchor=south,yshift=-1mm]

\tikzstyle{braceedge}=[decorate,decoration={brace,amplitude=2mm,raise=-1mm}]
\tikzstyle{small braceedge}=[decorate,decoration={brace,amplitude=1mm,raise=-1mm}]

\tikzstyle{doubled}=[line width=1.6pt] % set the line width for all doubled (quantum) maps/wires
\tikzstyle{boldedge}=[doubled,shorten <=-0.17mm,shorten >=-0.17mm]
\tikzstyle{boldedgegray}=[doubled,gray,shorten <=-0.17mm,shorten >=-0.17mm]
\tikzstyle{singleedgegray}=[gray]%,shorten <=-0.1mm,shorten >=-0.1mm]

\tikzstyle{semidoubled}=[line width=1.4pt] % set the line width for all doubled (quantum) maps/wires
\tikzstyle{semiboldedgegray}=[semidoubled,gray,shorten <=-0.17mm,shorten >=-0.17mm]

\tikzstyle{boxedge}=[semiboldedgegray]

\tikzstyle{boldedgedashed}=[very thick,dashed,shorten <=-0.17mm,shorten >=-0.17mm]
\tikzstyle{vboldedgedashed}=[doubled,dashed,shorten <=-0.17mm,shorten >=-0.17mm]
\tikzstyle{left hook arrow}=[left hook-latex]
\tikzstyle{right hook arrow}=[right hook-latex]
\tikzstyle{sembracket}=[line width=0.5pt,shorten <=-0.07mm,shorten >=-0.07mm]

\tikzstyle{causal edge}=[->,thick,gray]
\tikzstyle{causal nondir}=[thick,gray]
\tikzstyle{timeline}=[thick,gray, dashed]

\tikzstyle{cedge}=[<->,thick,gray!70!white]

\tikzstyle{empty diagram}=[draw=gray!40!white,dashed,shape=rectangle,minimum width=1cm,minimum height=1cm]
\tikzstyle{empty diagram small}=[draw=gray!50!white,dashed,shape=rectangle,minimum width=0.6cm,minimum height=0.5cm]

\tikzstyle{dot}=[inner sep=0mm,minimum width=2mm,minimum height=2mm,draw,shape=circle]
\tikzstyle{leak}=[white dot, shape=regular polygon, minimum size=3.3 mm, regular polygon sides=3, outer sep=-0.2mm, regular polygon rotate=270]
\tikzstyle{proj}=[draw,fill=white,chamfered rectangle,chamfered rectangle angle=30, minimum width=2mm, minimum height= 1mm,scale=0.5, outer sep=-0.2mm]
\tikzstyle{wide proj}=[draw,fill=white,chamfered rectangle,chamfered rectangle angle=30, minimum width=15mm,minimum height=1mm,scale=0.5, outer sep=-0.2mm]
\tikzstyle{very wide proj}=[draw,fill=white,chamfered rectangle,chamfered rectangle angle=30, minimum width=25mm,minimum height=1mm,scale=0.5, outer sep=-0.2mm]
\tikzstyle{very very wide proj}=[draw,fill=white,chamfered rectangle,chamfered rectangle angle=30, minimum width=35mm,minimum height=1mm,scale=0.5, outer sep=-0.2mm]
\tikzstyle{preleak}=[proj]

\tikzstyle{split proj out}=[regular polygon,regular polygon sides=3,draw,scale=0.75,inner sep=-0.5pt,minimum width=3.3mm,fill=white,regular polygon rotate=180]
\tikzstyle{split proj in}=[regular polygon,regular polygon sides=3,draw,scale=0.75,inner sep=-0.5pt,minimum width=3.3mm,fill=white]
\tikzstyle{Vleak}=[white dot, shape=regular polygon, minimum size=3.3 mm, regular polygon sides=3, outer sep=-0.2mm, regular polygon rotate=90]
\tikzstyle{dleak}=[white dot, line width=1.6pt, shape=regular polygon, minimum size=3.3 mm, regular polygon sides=3, outer sep=-0.2mm, regular polygon rotate=270]

\tikzstyle{Wsquare}=[white dot, shape=regular polygon, rounded corners=0.8 mm, minimum size=3.3 mm, regular polygon sides=3, outer sep=-0.2mm]
\tikzstyle{Wsquareadj}=[white dot, shape=regular polygon, rounded corners=0.8 mm, minimum size=3.3 mm, regular polygon sides=3, outer sep=-0.2mm, regular polygon rotate=180]
\tikzstyle{ddot}=[inner sep=0mm, doubled, minimum width=2.5mm,minimum height=2.5mm,draw,shape=circle]

\tikzstyle{black dot}=[dot,fill=black]
\tikzstyle{white dot}=[dot,fill=white,,text depth=-0.2mm]
\tikzstyle{white Wsquare}=[Wsquare,fill=gray,,text depth=-0.2mm]
\tikzstyle{white Wsquareadj}=[Wsquareadj,fill=white,,text depth=-0.2mm]
\tikzstyle{green dot}=[white dot] % for backwards-compatibility
\tikzstyle{gray dot}=[dot,fill=gray!40!white,,text depth=-0.2mm]
\tikzstyle{red dot}=[gray dot] % for backwards-compatibility

\tikzstyle{black ddot}=[ddot,fill=black]
\tikzstyle{white ddot}=[ddot,fill=white]
\tikzstyle{gray ddot}=[ddot,fill=gray!40!white]

\tikzstyle{gray edge}=[gray!60!white]

\tikzstyle{small dot}=[inner sep=0.5mm,minimum width=0pt,minimum height=0pt,draw,shape=circle]

\tikzstyle{small black dot}=[small dot,fill=black]
\tikzstyle{small white dot}=[small dot,fill=white]
\tikzstyle{small gray dot}=[small dot,fill=gray!40!white]

\tikzstyle{causal dot}=[inner sep=0.4mm,minimum width=0pt,minimum height=0pt,draw=white,shape=circle,fill=gray!40!white]

\tikzstyle{phase dimensions}=[minimum size=5mm,font=\footnotesize,rectangle,rounded corners=2.5mm,inner sep=0.2mm,outer sep=-2mm]
\tikzstyle{dphase dimensions}=[minimum size=5mm,font=\footnotesize,rectangle,rounded corners=2.5mm,inner sep=0.2mm,outer sep=-2mm]
\tikzstyle{white phase dot}=[dot,fill=white,phase dimensions]
\tikzstyle{white phase ddot}=[ddot,fill=white,dphase dimensions]

\tikzstyle{white rect ddot}=[draw=black,fill=white,doubled,minimum size=5mm,font=\footnotesize,rectangle,rounded corners=2.5mm,inner sep=0.2mm]
\tikzstyle{gray rect ddot}=[draw=black,fill=gray!40!white,doubled,minimum size=6mm,font=\footnotesize,rectangle,rounded corners=3mm]

\tikzstyle{gray phase dot}=[dot,fill=gray!40!white,phase dimensions]
\tikzstyle{gray phase ddot}=[ddot,fill=gray!40!white,dphase dimensions]
\tikzstyle{grey phase dot}=[gray phase dot]
\tikzstyle{grey phase ddot}=[gray phase ddot]

\tikzstyle{small phase dimensions}=[minimum size=4mm,font=\tiny,rectangle,rounded corners=2mm,inner sep=0.2mm,outer sep=-2mm]
\tikzstyle{small dphase dimensions}=[minimum size=4mm,font=\tiny,rectangle,rounded corners=2mm,inner sep=0.2mm,outer sep=-2mm]

\tikzstyle{small gray phase dot}=[dot,fill=gray!40!white,small phase dimensions]
\tikzstyle{small gray phase ddot}=[ddot,fill=gray!40!white,small dphase dimensions]

\tikzstyle{small map}=[draw,shape=rectangle,minimum height=4mm,minimum width=4mm,fill=white]

\tikzstyle{cnot}=[fill=white,shape=circle,inner sep=-1.4pt]

\tikzstyle{asym hadamard}=[fill=white,draw,shape=NEbox,inner sep=0.6mm,font=\footnotesize,minimum height=4mm]
\tikzstyle{asym hadamard conj}=[fill=white,draw,shape=NWbox,inner sep=0.6mm,font=\footnotesize,minimum height=4mm]
\tikzstyle{asym hadamard dag}=[fill=white,draw,shape=SEbox,inner sep=0.6mm,font=\footnotesize,minimum height=4mm]

\tikzstyle{hadamard}=[fill=white,draw,inner sep=0.6mm,font=\footnotesize,minimum height=4mm,minimum width=4mm]
\tikzstyle{small hadamard}=[fill=white,draw,inner sep=0.6mm,minimum height=1.5mm,minimum width=1.5mm]
\tikzstyle{small hadamard rotate}=[small hadamard,rotate=45]
\tikzstyle{dhadamard}=[hadamard,doubled]
\tikzstyle{small dhadamard}=[small hadamard,doubled]
\tikzstyle{small dhadamard rotate}=[small hadamard rotate,doubled]
\tikzstyle{antipode}=[white dot,inner sep=0.3mm,font=\footnotesize]

\tikzstyle{scalar}=[diamond,draw,inner sep=0.5pt,font=\small]
\tikzstyle{dscalar}=[diamond,doubled, draw,inner sep=0.5pt,font=\small]

\tikzstyle{small box}=[rectangle,inline text,fill=white,draw,minimum height=5mm,yshift=-0.5mm,minimum width=5mm,font=\small]
\tikzstyle{small gray box}=[small box,fill=gray!30]
\tikzstyle{medium box}=[rectangle,inline text,fill=white,draw,minimum height=5mm,yshift=-0.5mm,minimum width=10mm,font=\small]
\tikzstyle{square box}=[small box] % for backwards-compatibility
\tikzstyle{medium gray box}=[small box,fill=gray!30]
\tikzstyle{semilarge box}=[rectangle,inline text,fill=white,draw,minimum height=5mm,yshift=-0.5mm,minimum width=12.5mm,font=\small]
\tikzstyle{large box}=[rectangle,inline text,fill=white,draw,minimum height=5mm,yshift=-0.5mm,minimum width=15mm,font=\small]
\tikzstyle{large gray box}=[small box,fill=gray!30]

\tikzstyle{Bayes box}=[rectangle,fill=black,draw, minimum height=3mm, minimum width=3mm]

\tikzstyle{gray square point}=[small box,fill=gray!50]

\tikzstyle{dphase box white}=[dhadamard]
\tikzstyle{dphase box gray}=[dhadamard,fill=gray!50!white]
\tikzstyle{phase box white}=[hadamard]
\tikzstyle{phase box gray}=[hadamard,fill=gray!50!white]

\tikzstyle{point}=[regular polygon,regular polygon sides=3,draw,scale=0.75,inner sep=-0.5pt,minimum width=9mm,fill=white,regular polygon rotate=180]
\tikzstyle{point nosep}=[regular polygon,regular polygon sides=3,draw,scale=0.75,inner sep=-2pt,minimum width=9mm,fill=white,regular polygon rotate=180]
\tikzstyle{copoint}=[regular polygon,regular polygon sides=3,draw,scale=0.75,inner sep=-0.5pt,minimum width=9mm,fill=white]
\tikzstyle{dpoint}=[point,doubled]
\tikzstyle{dcopoint}=[copoint,doubled]

\tikzstyle{pointgrow}=[shape=cornerpoint,kpoint common,scale=0.75,inner sep=3pt]
\tikzstyle{pointgrow dag}=[shape=cornercopoint,kpoint common,scale=0.75,inner sep=3pt]

\tikzstyle{wide copoint}=[fill=white,draw,shape=isosceles triangle,shape border rotate=90,isosceles triangle stretches=true,inner sep=0pt,minimum width=1.5cm,minimum height=6.12mm]
\tikzstyle{wide point}=[fill=white,draw,shape=isosceles triangle,shape border rotate=-90,isosceles triangle stretches=true,inner sep=0pt,minimum width=1.5cm,minimum height=6.12mm,yshift=-0.0mm]
\tikzstyle{wide point plus}=[fill=white,draw,shape=isosceles triangle,shape border rotate=-90,isosceles triangle stretches=true,inner sep=0pt,minimum width=1.74cm,minimum height=7mm,yshift=-0.0mm]

\tikzstyle{wide dpoint}=[fill=white,doubled,draw,shape=isosceles triangle,shape border rotate=-90,isosceles triangle stretches=true,inner sep=0pt,minimum width=1.5cm,minimum height=6.12mm,yshift=-0.0mm]

\tikzstyle{tinypoint}=[regular polygon,regular polygon sides=3,draw,scale=0.55,inner sep=-0.15pt,minimum width=6mm,fill=white,regular polygon rotate=180]

\tikzstyle{white point}=[point]
\tikzstyle{white dpoint}=[dpoint]
\tikzstyle{green point}=[white point] % for backwards-compatibility
\tikzstyle{white copoint}=[copoint]
\tikzstyle{gray point}=[point,fill=gray!40!white]
\tikzstyle{gray dpoint}=[gray point,doubled]
\tikzstyle{red point}=[gray point] % for backwards-compatibility
\tikzstyle{gray copoint}=[copoint,fill=gray!40!white]
\tikzstyle{gray dcopoint}=[gray copoint,doubled]

\tikzstyle{white point guide}=[regular polygon,regular polygon sides=3,font=\scriptsize,draw,scale=0.65,inner sep=-0.5pt,minimum width=9mm,fill=white,regular polygon rotate=180]

\tikzstyle{black point}=[point,fill=black,font=\color{white}]
\tikzstyle{black copoint}=[copoint,fill=black,font=\color{white}]

\tikzstyle{tiny gray point}=[tinypoint,fill=gray!40!white]

\tikzstyle{diredge}=[->]
\tikzstyle{ddiredge}=[<->]
\tikzstyle{rdiredge}=[<-]
\tikzstyle{thickdiredge}=[->, very thick]
\tikzstyle{pointer edge}=[->,very thick,gray]
\tikzstyle{pointer edge part}=[very thick,gray]
\tikzstyle{dashed edge}=[dashed]
\tikzstyle{thick dashed edge}=[very thick,dashed]
\tikzstyle{thick gray dashed edge}=[thick dashed edge,gray!40]
\tikzstyle{thick map edge}=[very thick,|->]

\makeatletter
\newcommand{\boxshape}[3]{%
\pgfdeclareshape{#1}{
\inheritsavedanchors[from=rectangle] % this is nearly a rectangle
\inheritanchorborder[from=rectangle]
\inheritanchor[from=rectangle]{center}
\inheritanchor[from=rectangle]{north}
\inheritanchor[from=rectangle]{south}
\inheritanchor[from=rectangle]{west}
\inheritanchor[from=rectangle]{east}
\backgroundpath{% this is new
\southwest \pgf@xa=\pgf@x \pgf@ya=\pgf@y
\northeast \pgf@xb=\pgf@x \pgf@yb=\pgf@y

\@tempdima=#2
\@tempdimb=#3

\pgfpathmoveto{\pgfpoint{\pgf@xa - 5pt + \@tempdima}{\pgf@ya}}
\pgfpathlineto{\pgfpoint{\pgf@xa - 5pt - \@tempdima}{\pgf@yb}}
\pgfpathlineto{\pgfpoint{\pgf@xb + 5pt + \@tempdimb}{\pgf@yb}}
\pgfpathlineto{\pgfpoint{\pgf@xb + 5pt - \@tempdimb}{\pgf@ya}}
\pgfpathlineto{\pgfpoint{\pgf@xa - 5pt + \@tempdima}{\pgf@ya}}
\pgfpathclose
}
}}

\boxshape{NEbox}{0pt}{3pt}
\boxshape{SEbox}{0pt}{-3pt}
\boxshape{NWbox}{5pt}{0pt}
\boxshape{SWbox}{-5pt}{0pt}
\boxshape{EBox}{-3pt}{3pt}
\boxshape{WBox}{3pt}{-3pt}
\makeatother

\tikzstyle{cloud}=[shape=cloud,draw,minimum width=1.5cm,minimum height=1.5cm]

\tikzstyle{map}=[draw,shape=NEbox,inner sep=1pt,minimum height=4mm,fill=white]
\tikzstyle{dashedmap}=[draw,dashed,shape=NEbox,inner sep=2pt,minimum height=6mm,fill=white]
\tikzstyle{mapdag}=[draw,shape=SEbox,inner sep=1pt,minimum height=4mm,fill=white]
\tikzstyle{mapadj}=[draw,shape=SEbox,inner sep=2pt,minimum height=6mm,fill=white]
\tikzstyle{maptrans}=[draw,shape=SWbox,inner sep=2pt,minimum height=6mm,fill=white]
\tikzstyle{mapconj}=[draw,shape=NWbox,inner sep=2pt,minimum height=6mm,fill=white]

\tikzstyle{medium map}=[draw,shape=NEbox,inner sep=2pt,minimum height=6mm,fill=white,minimum width=7mm]
\tikzstyle{medium map dag}=[draw,shape=SEbox,inner sep=2pt,minimum height=6mm,fill=white,minimum width=7mm]
\tikzstyle{medium map adj}=[draw,shape=SEbox,inner sep=2pt,minimum height=6mm,fill=white,minimum width=7mm]
\tikzstyle{medium map trans}=[draw,shape=SWbox,inner sep=2pt,minimum height=6mm,fill=white,minimum width=7mm]
\tikzstyle{medium map conj}=[draw,shape=NWbox,inner sep=2pt,minimum height=6mm,fill=white,minimum width=7mm]
\tikzstyle{semilarge map}=[draw,shape=NEbox,inner sep=2pt,minimum height=6mm,fill=white,minimum width=9.5mm]
\tikzstyle{semilarge map trans}=[draw,shape=SWbox,inner sep=2pt,minimum height=6mm,fill=white,minimum width=9.5mm]
\tikzstyle{semilarge map adj}=[draw,shape=SEbox,inner sep=2pt,minimum height=6mm,fill=white,minimum width=9.5mm]
\tikzstyle{semilarge map dag}=[draw,shape=SEbox,inner sep=2pt,minimum height=6mm,fill=white,minimum width=9.5mm]
\tikzstyle{semilarge map conj}=[draw,shape=NWbox,inner sep=2pt,minimum height=6mm,fill=white,minimum width=9.5mm]
\tikzstyle{large map}=[draw,shape=NEbox,inner sep=2pt,minimum height=6mm,fill=white,minimum width=12mm]
\tikzstyle{large map conj}=[draw,shape=NWbox,inner sep=2pt,minimum height=6mm,fill=white,minimum width=12mm]
\tikzstyle{very large map}=[draw,shape=NEbox,inner sep=2pt,minimum height=6mm,fill=white,minimum width=17mm]

\tikzstyle{medium dmap}=[draw,doubled,shape=NEbox,inner sep=2pt,minimum height=6mm,fill=white,minimum width=7mm]
\tikzstyle{medium dmap dag}=[draw,doubled,shape=SEbox,inner sep=2pt,minimum height=6mm,fill=white,minimum width=7mm]
\tikzstyle{medium dmap adj}=[draw,doubled,shape=SEbox,inner sep=2pt,minimum height=6mm,fill=white,minimum width=7mm]
\tikzstyle{medium dmap trans}=[draw,doubled,shape=SWbox,inner sep=2pt,minimum height=6mm,fill=white,minimum width=7mm]
\tikzstyle{medium dmap conj}=[draw,doubled,shape=NWbox,inner sep=2pt,minimum height=6mm,fill=white,minimum width=7mm]
\tikzstyle{semilarge dmap}=[draw,doubled,shape=NEbox,inner sep=2pt,minimum height=6mm,fill=white,minimum width=9.5mm]
\tikzstyle{semilarge dmap trans}=[draw,doubled,shape=SWbox,inner sep=2pt,minimum height=6mm,fill=white,minimum width=9.5mm]
\tikzstyle{semilarge dmap adj}=[draw,doubled,shape=SEbox,inner sep=2pt,minimum height=6mm,fill=white,minimum width=9.5mm]
\tikzstyle{semilarge dmap dag}=[draw,doubled,shape=SEbox,inner sep=2pt,minimum height=6mm,fill=white,minimum width=9.5mm]
\tikzstyle{semilarge dmap conj}=[draw,doubled,shape=NWbox,inner sep=2pt,minimum height=6mm,fill=white,minimum width=9.5mm]
\tikzstyle{large dmap}=[draw,doubled,shape=NEbox,inner sep=2pt,minimum height=6mm,fill=white,minimum width=12mm]
\tikzstyle{large dmap conj}=[draw,doubled,shape=NWbox,inner sep=2pt,minimum height=6mm,fill=white,minimum width=12mm]
\tikzstyle{large dmap trans}=[draw,doubled,shape=SWbox,inner sep=2pt,minimum height=6mm,fill=white,minimum width=12mm]
\tikzstyle{large dmap adj}=[draw,doubled,shape=SEbox,inner sep=2pt,minimum height=6mm,fill=white,minimum width=12mm]
\tikzstyle{large dmap dag}=[draw,doubled,shape=SEbox,inner sep=2pt,minimum height=6mm,fill=white,minimum width=12mm]
\tikzstyle{very large dmap}=[draw,doubled,shape=NEbox,inner sep=2pt,minimum height=6mm,fill=white,minimum width=19.5mm]

\tikzstyle{muxbox}=[draw,shape=rectangle,minimum height=3mm,minimum width=3mm,fill=white]
\tikzstyle{dmuxbox}=[muxbox,doubled]

\tikzstyle{box}=[draw,shape=rectangle,inner sep=2pt,minimum height=6mm,minimum width=6mm,fill=white]
\tikzstyle{dbox}=[draw,doubled,shape=rectangle,inner sep=2pt,minimum height=6mm,minimum width=6mm,fill=white]
\tikzstyle{dmap}=[draw,doubled,shape=NEbox,inner sep=2pt,minimum height=6mm,fill=white]
\tikzstyle{dmapdag}=[draw,doubled,shape=SEbox,inner sep=2pt,minimum height=6mm,fill=white]
\tikzstyle{dmapadj}=[draw,doubled,shape=SEbox,inner sep=2pt,minimum height=6mm,fill=white]
\tikzstyle{dmaptrans}=[draw,doubled,shape=SWbox,inner sep=2pt,minimum height=6mm,fill=white]
\tikzstyle{dmapconj}=[draw,doubled,shape=NWbox,inner sep=2pt,minimum height=6mm,fill=white]

\tikzstyle{ddmap}=[draw,doubled,dashed,shape=NEbox,inner sep=2pt,minimum height=6mm,fill=white]
\tikzstyle{ddmapdag}=[draw,doubled,dashed,shape=SEbox,inner sep=2pt,minimum height=6mm,fill=white]
\tikzstyle{ddmapadj}=[draw,doubled,dashed,shape=SEbox,inner sep=2pt,minimum height=6mm,fill=white]
\tikzstyle{ddmaptrans}=[draw,doubled,dashed,shape=SWbox,inner sep=2pt,minimum height=6mm,fill=white]
\tikzstyle{ddmapconj}=[draw,doubled,dashed,shape=NWbox,inner sep=2pt,minimum height=6mm,fill=white]

\boxshape{sNEbox}{0pt}{3pt}
\boxshape{sSEbox}{0pt}{-3pt}
\boxshape{sNWbox}{3pt}{0pt}
\boxshape{sSWbox}{-3pt}{0pt}
\tikzstyle{smap}=[draw,shape=sNEbox,fill=white]
\tikzstyle{smapdag}=[draw,shape=sSEbox,fill=white]
\tikzstyle{smapadj}=[draw,shape=sSEbox,fill=white]
\tikzstyle{smaptrans}=[draw,shape=sSWbox,fill=white]
\tikzstyle{smapconj}=[draw,shape=sNWbox,fill=white]

\tikzstyle{dsmap}=[draw,dashed,shape=sNEbox,fill=white]
\tikzstyle{dsmapdag}=[draw,dashed,shape=sSEbox,fill=white]
\tikzstyle{dsmaptrans}=[draw,dashed,shape=sSWbox,fill=white]
\tikzstyle{dsmapconj}=[draw,dashed,shape=sNWbox,fill=white]

\boxshape{mNEbox}{0pt}{10pt}
\boxshape{mSEbox}{0pt}{-10pt}
\boxshape{mNWbox}{10pt}{0pt}
\boxshape{mSWbox}{-10pt}{0pt}
\tikzstyle{mmap}=[draw,shape=mNEbox]
\tikzstyle{mmapdag}=[draw,shape=mSEbox]
\tikzstyle{mmaptrans}=[draw,shape=mSWbox]
\tikzstyle{mmapconj}=[draw,shape=mNWbox]

\tikzstyle{mmapgray}=[draw,fill=gray!40!white,shape=mNEbox]
\tikzstyle{smapgray}=[draw,fill=gray!40!white,shape=sNEbox]

\makeatletter

\pgfdeclareshape{cornerpoint}{
\inheritsavedanchors[from=rectangle] % this is nearly a rectangle
\inheritanchorborder[from=rectangle]
\inheritanchor[from=rectangle]{center}
\inheritanchor[from=rectangle]{north}
\inheritanchor[from=rectangle]{south}
\inheritanchor[from=rectangle]{west}
\inheritanchor[from=rectangle]{east}
\backgroundpath{% this is new
\southwest \pgf@xa=\pgf@x \pgf@ya=\pgf@y
\northeast \pgf@xb=\pgf@x \pgf@yb=\pgf@y

\pgfmathsetmacro{\pgf@shorten@left}{\pgfkeysvalueof{/tikz/shorten left}}
\pgfmathsetmacro{\pgf@shorten@right}{\pgfkeysvalueof{/tikz/shorten right}}

\pgfpathmoveto{\pgfpoint{0.5 * (\pgf@xa + \pgf@xb)}{\pgf@ya - 5pt}}
\pgfpathlineto{\pgfpoint{\pgf@xa - 8pt + \pgf@shorten@left}{\pgf@yb - 1.5 * \pgf@shorten@left}}
\pgfpathlineto{\pgfpoint{\pgf@xa - 8pt + \pgf@shorten@left}{\pgf@yb}}
\pgfpathlineto{\pgfpoint{\pgf@xb + 8pt - \pgf@shorten@right}{\pgf@yb}}
\pgfpathlineto{\pgfpoint{\pgf@xb + 8pt - \pgf@shorten@right}{\pgf@yb - 1.5 * \pgf@shorten@right}}
\pgfpathclose
}
}

\pgfdeclareshape{cornercopoint}{
\inheritsavedanchors[from=rectangle] % this is nearly a rectangle
\inheritanchorborder[from=rectangle]
\inheritanchor[from=rectangle]{center}
\inheritanchor[from=rectangle]{north}
\inheritanchor[from=rectangle]{south}
\inheritanchor[from=rectangle]{west}
\inheritanchor[from=rectangle]{east}
\backgroundpath{% this is new
\southwest \pgf@xa=\pgf@x \pgf@ya=\pgf@y
\northeast \pgf@xb=\pgf@x \pgf@yb=\pgf@y

\pgfmathsetmacro{\pgf@shorten@left}{\pgfkeysvalueof{/tikz/shorten left}}
\pgfmathsetmacro{\pgf@shorten@right}{\pgfkeysvalueof{/tikz/shorten right}}

\pgfpathmoveto{\pgfpoint{0.5 * (\pgf@xa + \pgf@xb)}{\pgf@yb + 5pt}}
\pgfpathlineto{\pgfpoint{\pgf@xa - 8pt + \pgf@shorten@left}{\pgf@ya + 1.5 * \pgf@shorten@left}}
\pgfpathlineto{\pgfpoint{\pgf@xa - 8pt + \pgf@shorten@left}{\pgf@ya}}
\pgfpathlineto{\pgfpoint{\pgf@xb + 8pt - \pgf@shorten@right}{\pgf@ya}}
\pgfpathlineto{\pgfpoint{\pgf@xb + 8pt - \pgf@shorten@right}{\pgf@ya + 1.5 * \pgf@shorten@right}}
\pgfpathclose
}
}

\makeatother

\pgfkeyssetvalue{/tikz/shorten left}{0pt}
\pgfkeyssetvalue{/tikz/shorten right}{0pt}

\tikzstyle{kpoint common}=[draw,fill=white,inner sep=1pt,minimum height=4mm]
\tikzstyle{kpoint sc}=[shape=cornerpoint,kpoint common]
\tikzstyle{kpoint adjoint sc}=[shape=cornercopoint,kpoint common]
\tikzstyle{kpoint}=[shape=cornerpoint,shorten left=5pt,kpoint common]
\tikzstyle{kpoint adjoint}=[shape=cornercopoint,shorten left=5pt,kpoint common]
\tikzstyle{kpoint conjugate}=[shape=cornerpoint,shorten right=5pt,kpoint common]
\tikzstyle{kpoint transpose}=[shape=cornercopoint,shorten right=5pt,kpoint common]
\tikzstyle{kpoint symm}=[shape=cornerpoint,shorten left=5pt,shorten right=5pt,kpoint common]

\tikzstyle{wide kpoint sc}=[shape=cornerpoint,kpoint common, minimum width=1 cm]
\tikzstyle{wide kpointdag sc}=[shape=cornercopoint,kpoint common, minimum width=1 cm]

\tikzstyle{black kpoint}=[shape=cornerpoint,shorten left=5pt,kpoint common,fill=black,font=\color{white}]

\tikzstyle{black kpoint sm}=[shape=cornerpoint,shorten left=5pt,kpoint common,fill=black,font=\color{white},scale=0.75]

\tikzstyle{black kpoint adjoint}=[shape=cornercopoint,shorten left=5pt,kpoint common,fill=black,font=\color{white}]
\tikzstyle{black kpointadj}=[shape=cornercopoint,shorten left=5pt,kpoint common,fill=black,font=\color{white}]

\tikzstyle{black kpointadj sm}=[shape=cornercopoint,shorten left=5pt,kpoint common,fill=black,font=\color{white},scale=0.75]

\tikzstyle{black dkpoint}=[shape=cornerpoint,shorten left=5pt,kpoint common,fill=black, doubled,font=\color{white}]
\tikzstyle{black dkpoint adjoint}=[shape=cornercopoint,shorten left=5pt,kpoint common,fill=black, doubled,font=\color{white}]
\tikzstyle{black dkpointadj}=[shape=cornercopoint,shorten left=5pt,kpoint common,fill=black, doubled,font=\color{white}]

\tikzstyle{black dkpoint sm}=[shape=cornerpoint,shorten left=5pt,kpoint common,fill=black, doubled,font=\color{white},scale=0.75]
\tikzstyle{black dkpointadj sm}=[shape=cornercopoint,shorten left=5pt,kpoint common,fill=black, doubled,font=\color{white},scale=0.75]

\tikzstyle{kpointdag}=[kpoint adjoint]
\tikzstyle{kpointadj}=[kpoint adjoint]
\tikzstyle{kpointconj}=[kpoint conjugate]
\tikzstyle{kpointtrans}=[kpoint transpose]

\tikzstyle{big kpoint}=[kpoint, minimum width=1.2 cm, minimum height=8mm, inner sep=4pt, text depth=3mm]

\tikzstyle{wide kpoint}=[kpoint, minimum width=1 cm, inner sep=2pt]%, text depth=-0.7 mm]
\tikzstyle{wide kpointdag}=[kpointdag, minimum width=1 cm, inner sep=2pt]%, text depth=0.7 mm]
\tikzstyle{wide kpointconj}=[kpointconj, minimum width=1 cm, inner sep=2pt]%, text depth=-0.7 mm]
\tikzstyle{wide kpointtrans}=[kpointtrans, minimum width=1 cm, inner sep=2pt]%, text depth=0.7 mm]

\tikzstyle{wider kpoint}=[kpoint, minimum width=1.25 cm, inner sep=2pt]%, text depth=-0.7 mm]
\tikzstyle{wider kpointdag}=[kpointdag, minimum width=1.25 cm, inner sep=2pt]%, text depth=0.7 mm]
\tikzstyle{wider kpointconj}=[kpointconj, minimum width=1.25 cm, inner sep=2pt]%, text depth=-0.7 mm]
\tikzstyle{wider kpointtrans}=[kpointtrans, minimum width=1.25 cm, inner sep=2pt]%, text depth=0.7 mm]

\tikzstyle{gray kpoint}=[kpoint,fill=gray!50!white]
\tikzstyle{gray kpointdag}=[kpointdag,fill=gray!50!white]
\tikzstyle{gray kpointadj}=[kpointadj,fill=gray!50!white]
\tikzstyle{gray kpointconj}=[kpointconj,fill=gray!50!white]
\tikzstyle{gray kpointtrans}=[kpointtrans,fill=gray!50!white]

\tikzstyle{gray dkpoint}=[kpoint,fill=gray!50!white,doubled]
\tikzstyle{gray dkpointdag}=[kpointdag,fill=gray!50!white,doubled]
\tikzstyle{gray dkpointadj}=[kpointadj,fill=gray!50!white,doubled]
\tikzstyle{gray dkpointconj}=[kpointconj,fill=gray!50!white,doubled]
\tikzstyle{gray dkpointtrans}=[kpointtrans,fill=gray!50!white,doubled]

\tikzstyle{white label}=[draw,fill=white,rectangle,inner sep=0.7 mm]
\tikzstyle{gray label}=[draw,fill=gray!50!white,rectangle,inner sep=0.7 mm]
\tikzstyle{black label}=[draw,fill=black,rectangle,inner sep=0.7 mm]

\tikzstyle{dkpoint}=[kpoint,doubled]
\tikzstyle{wide dkpoint}=[wide kpoint,doubled]
\tikzstyle{dkpointdag}=[kpoint adjoint,doubled]
\tikzstyle{wide dkpointdag}=[wide kpointdag,doubled]
\tikzstyle{dkcopoint}=[kpoint adjoint,doubled]
\tikzstyle{dkpointadj}=[kpoint adjoint,doubled]
\tikzstyle{dkpointconj}=[kpoint conjugate,doubled]
\tikzstyle{dkpointtrans}=[kpoint transpose,doubled]

\tikzstyle{kscalar}=[kpoint common, shape=EBox, inner xsep=-1pt, inner ysep=3pt,font=\small]
\tikzstyle{kscalarconj}=[kpoint common, shape=WBox, inner xsep=-1pt, inner ysep=3pt,font=\small]

\tikzstyle{spekpoint}=[kpoint sc,minimum height=5mm,inner sep=3pt]
\tikzstyle{spekcopoint}=[kpoint adjoint sc,minimum height=5mm,inner sep=3pt]

\tikzstyle{dspekpoint}=[spekpoint,doubled]
\tikzstyle{dspekcopoint}=[spekcopoint,doubled]

 \tikzstyle{upground}=[circuit ee IEC,thick,ground,rotate=90,scale=2.5]
 \tikzstyle{downground}=[circuit ee IEC,thick,ground,rotate=-90,scale=2.5]
 \tikzstyle{bigground}=[regular polygon,regular polygon sides=3,draw=gray,scale=0.50,inner sep=-0.5pt,minimum width=10mm,fill=gray]
\tikzstyle{arrs}=[-latex,font=\small,auto]
\tikzstyle{arrow plain}=[arrs]
\tikzstyle{arrow dashed}=[dashed,arrs]
\tikzstyle{arrow bold}=[very thick,arrs]
\tikzstyle{arrow hide}=[draw=white!0,-]
\tikzstyle{arrow reverse}=[latex-]
\tikzstyle{cdnode}=[]

\let\olddagger\dagger
\renewcommand{\dagger}{\ensuremath{\olddagger}\xspace}

\theoremstyle{definition}

\newtheorem{example*}[theorem]{Example*}
\newtheorem{examples*}[theorem]{Examples*}

\newtheorem{remark*}[theorem]{Remark*}

\theoremstyle{plain}

\usepackage{color}
\def\bR{\begin{color}{red}}
\def\bB{\begin{color}{blue}}
\def\bM{\begin{color}{magenta}}
\def\bC{\begin{color}{cyan}}
\def\bW{\begin{color}{white}}
\def\bBl{\begin{color}{black}}
\def\bG{\begin{color}{green}}
\def\bY{\begin{color}{yellow}}
\def\e{\end{color}\xspace}
\newcommand{\bit}{\begin{itemize}}
\newcommand{\eit}{\end{itemize}\par\noindent}
\newcommand{\ben}{\begin{enumerate}}
\newcommand{\een}{\end{enumerate}\par\noindent}
\newcommand{\beq}{\begin{equation}}
\newcommand{\eeq}{\end{equation}\par\noindent}
\newcommand{\beqa}{\begin{eqnarray*}}
\newcommand{\eeqa}{\end{eqnarray*}\par\noindent}
\newcommand{\beqn}{\begin{eqnarray}}
\newcommand{\eeqn}{\end{eqnarray}\par\noindent}
\def\jR{\begin{color}{DarkRed}}
\def\jB{\begin{color}{MidnightBlue}}
\def\jM{\begin{color}{magenta}}
\def\jC{\begin{color}{cyan}}
\def\jW{\begin{color}{white}}
\def\jBl{\begin{color}{black}}
\def\jG{\begin{color}{green}}
\def\jY{\begin{color}{yellow}}

\def\cR{\begin{color}{Crimson}}
\def\cB{\begin{color}{cyan}}
\def\cM{\begin{color}{magenta}}
\def\cC{\begin{color}{cyan}}
\def\cW{\begin{color}{white}}
\def\cBl{\begin{color}{black}}
\def\cG{\begin{color}{green}}
\def\cY{\begin{color}{yellow}}
\begin{document}

\title{Activation of post-quantum steering}

\author{Ana Bel\'en Sainz}
\email{ana.sainz@ug.edu.pl}
\affiliation{International Centre for Theory of Quantum Technologies, University of  Gda{\'n}sk, 80-309 Gda{\'n}sk, Poland}
\affiliation{Basic Research Community for Physics e.V., Germany}
\author{Paul Skrzypczyk}
\affiliation{H.~H.~Wills Physics Laboratory, University of Bristol, Tyndall Avenue, Bristol, BS8 1TL, United Kingdom}
\affiliation{CIFAR Azrieli Global Scholars program, CIFAR, Toronto, Canada}
\author{Matty J.~Hoban}
\affiliation{Quantum Group, Department of Computer Science, University of Oxford, Wolfson Building, Parks Road, Oxford, OX1 3QD, United Kingdom}

\begin{abstract}
There are possible physical theories that give greater violations of Bell's inequalities than the corresponding Tsirelson bound, termed post-quantum non-locality. Such theories do not violate special relativity, but could give an advantage in certain information processing tasks. There is another way in which entangled quantum states exhibit non-classical phenomena, with one notable example being Einstein-Podolsky-Rosen (EPR) steering; a violation of a bipartite Bell inequality implies EPR steering, but the converse is not necessarily true. The study of post-quantum EPR steering is more intricate, but it has been shown that it does not always imply post-quantum non-locality in a conventional Bell test. In this work we show how to distribute resources in a larger network that individually do not demonstrate post-quantum non-locality but violate a Tsirelson bound for the network. That is, we show how to activate post-quantum steering so that it can now be witnessed as post-quantum correlations in a Bell scenario. One element of our work that may be of independent interest is we show how to self-test a bipartite quantum assemblage in a network, even assuming post-quantum resources.
\end{abstract}

\maketitle

In a Bell test, local measurements on entangled quantum states can produce statistics that cannot be simulated classically without communication.  This is usually witnessed by the violation of a Bell inequality \cite{Bell_1964,BellReview,scarani2019bell}. In such a setting, there also exist probability distributions that cannot be simulated \textit{quantumly} without communication, i.e., not even local measurements on any entangled quantum state can produce the statistics.  These (mathematically allowed) statistics that cannot be simulated quantumly are called \textit{post-quantum correlations}  \cite{popescu2014nonlocality}. Notably, there are post-quantum correlations that cannot be used to communicate: they satisfy the no-signalling principle \cite{popescu1994quantum}. This led to a line of research examining the consequences of such non-signalling, post-quantum correlations \cite{van2013implausible,brassard2006limit,navascues2010glance,pawlowski2009information,fritz2013local}.~Interestingly, one can examine these consequences independent of physical realisations of the correlations, in a largely theory-independent manner; this is called the \textit{device-independent approach} \cite{mayers1998quantum,barrett2005no,acin2007device}.

While it has been fruitful to study correlations in such a device-independent manner, especially with applications to cryptography in mind, it could be perceived as quite limiting. We can have situations involving multiple systems where the behaviour of one of these systems is well characterised and ignorance of the physical system is inappropriate. One such setting is that of an Einstein-Podolsky-Rosen (EPR) experiment, or scenario \cite{einstein1935can,wiseman2007steering,reid2009colloquium,uola2020quantum}. 
 Traditionally, in such an experiment, one party (Alice) who has an uncharacterised device, tries to convince Bob that they share systems in an entangled quantum state. Alice can successfully convince Bob if he witnesses the phenomenon of ``EPR-steering". While this is implied by the violation of a Bell inequality, it is not equivalent to it \cite{quintino2015inequivalence}: there are experiments that produce correlations which exhibit EPR-steering but do not violate any Bell inequality.  Quantum EPR-steering has been quite extensively studied, especially since its quantum-information reframing \cite{wiseman2007steering}, with some in-depth reviews and more recent results being captured in Refs.~\cite{reid2009colloquium,uola2020quantum,cavalcanti2016quantum,lee2025unveiling,bizzarri2024quantum,PhysRevA.110.012418}.

With the above in mind, Sainz\textit{ et.~al.}~asked whether there could ever exist demonstrations of EPR-steering which cannot be reproduced by local measurements on an entangled quantum state.  Inspired by the case of Bell scenarios, this type of EPR-steering that cannot be simulated quantumly was termed  ``post-quantum steering" \cite{sainz2015postquantum}. In the scenario that Ref.~\cite{sainz2015postquantum} studies,  one still assumes that Bob's device is quantum and well characterised, but Alice's device is not. This is in contrast to other options to study steering beyond quantum theory, where Bob's device is still characterised but now the states of his system are described according to a theory that goes beyond quantum theory (e.g., within the framework of generalised probabilistic theories (GPT) \cite{Hardy,GPT_Barrett}). While this GPT approach leads to interesting situations \cite{barnum2013ensemble,jenvcova2022assemblages}, the post-quantum approach pioneered by Ref.~\cite{sainz2015postquantum} has two main advantages. On the one hand, it enables one to explore the consequences of a locally-quantum world. On the other hand,  it opens the door to the question of how to certify that your experiment is locally-quantum yet not globally underpinned by a multipartite quantum state \cite{sainz2018formalism}.  
Coming back to the question of whether post-quantum steering exists for a scenario with a quantum Bob, it was well known that this isn't the case in the traditional setting (involving only Alice and Bob).  This is due to the Gisin-Hughston-Jozsa-Wootters (GHJW) theorem \cite{gisin1989stochastic,hughston1993complete}, which finds a quantum representation of any EPR-steering in such a bipartite scenario. However, as pointed out in Ref.~\cite{sainz2015postquantum}, if one considers generalised scenarios, involving, e.g.,  three separated parties where two of them use uncharacterised devices, then post-quantum steering is now possible. 
It was also later shown that other modifications to the standard EPR experiment can also accommodate post-quantum steering \cite{sainz2020bipartite,zjawin2022resource}; these include so-called Bob-with-Input scenarios and Instrumental EPR scenarios.  
Remarkably, Refs.~\cite{sainz2015postquantum,sainz2020bipartite,sainz2018formalism}  further showed that not only does post-quantum steering exist, but that it is distinct from post-quantum correlations: post-quantum steering can be exhibited \emph{without ever producing post-quantum correlations} in a naturally-associated Bell scenario. This indeed holds true for any choice or number of measurements that the parties (who were originally characterised in the EPR scenario but in this Bell scenario are now uncharacterised) could perform.

Since the publication of Ref.~\cite{sainz2015postquantum} subsequent work has been devoted to developing mathematical frameworks to explore the phenomenon of post-quantum steering \cite{hoban2018channel,sainz2018formalism,sainz2020bipartite}, understanding it as a resource \cite{cavalcanti2022post,Zjawin2023quantifyingepr,zjawin2022resource}, and unravelling the properties of the most general kind of resources that do not lead to signalling \cite{banacki2021edge,ramanathan2022single}. 
For instance, Ref.~\cite{zjawin2022resource} provides a unified framework for studying the nonclassicality of various generalizations of the EPR scenario. It  studies conversions between post-quantum resources (both analytically and numerically)  when the free operations are Local Operations and Shared Randomness. 
Moreover, Ref.~\cite{cavalcanti2022post} shows that post-quantum steering provides an advantage over quantum resources when the task of Remote State Preparation is considered: here, Alice is asked to deterministically prepare a system on state $\ket{\psi}$ in Bob's lab, having only a classical description of this state and trying to minimise the amount of classical communication she should send Bob. Post-quantum steering enables this communication to be only one bit when the system is a qubit, whereas quantum theory requires 2 bits. The search for a task with post-quantum steering advantage took place in Ref.~\cite{cavalcanti2022post} beyond what is considered device-independent. The reason for this was that post-quantum correlations don't necessarily capture post-quantum steering, as mentioned in the previous paragraph. 
 It would appear that the device-independent approach is thus limited for studying the consequences of post-quantum steering. 

In this work we show that this is not the full picture. We show that any resource that demonstrates post-quantum steering can be combined within a larger network of quantum resources to exhibit post-quantum correlations. In particular, multiple experiments that individually cannot exhibit post-quantum correlations can be combined in a network that does. We call this ``activation of post-quantum steering", analogous to the (super)activation of bound entanglement \cite{horodecki1999bound}.  Recall that in the latter, two different bound entangled states (those from which we cannot individually distill maximally entangled states) can be combined to become distillable.  The analogous resource  to a bound entangled state in our setting is post-quantum steering that does not demonstrate post-quantum correlations. 

To be more precise, we show how to activate tripartite post-quantum steering resources. To prove that it can be activated we require two components: 1) to place the tripartite resources in a setting with four parties, and 2) the ability to certify that certain quantum resources are prepared between two parties among the four. The  latter  is referred to as \textit{self-testing of assemblages}, as studied previously in  Ref.~\cite{chen2021robust}.  Here, however,  we generalise the setting to allow for the two parties to share post-quantum resources within the larger network of four or more parties. The bipartite resource that is self-tested is that which results from Pauli measurements made on a maximally entangled state;  such a resource can be used to perform ``remote quantum state tomography" \cite{bowles2018device}. That is, in an indirect way, we can certify (up to certain symmetries) that tomographically-complete measurements are being made on a quantum system, even without assuming anything about the quantum system; this is done in a device-independent manner, based on correlations. This then allows us to witness post-quantum steering in a device-independent way. 

\section{Einstein-Podolsky-Rosen scenarios}

There has been much literature on the study of EPR-steering, with numerous in-depth reviews, e.g.~Refs.~\cite{reid2009colloquium,uola2020quantum,cavalcanti2016quantum}. The aforementioned traditional EPR scenario is the one with two parties, Alice and Bob, where Alice's system is treated as a black box that takes in inputs and produces outputs, and Bob's system is assumed to be quantum.  Bob's system is thus associated with a Hilbert space of a known dimension, upon which he can make known measurements. Alice and Bob's systems can have shared resources such as randomness or entanglement. The setting relevant to this work primarily has three parties, Alice, Bob and Charlie. Now Alice and Bob have uncharacterised devices that are treated as black boxes; Charlie's system is now quantum, associated with a known Hilbert space, and upon which Charlie can make known measurements. The three parties can share some arbitrary resource, as long as they cannot use it to communicate to each other.

The important consequence of the above is that Charlie can perform full state tomography upon his quantum system, while Alice and Bob input classical data to their black boxes and receive  classical output data. See Fig.~\ref{fig:tradtrip} for a schematic of the setting. As can be seen, Alice's (Bob's) inputs $x\in\mathbb{X}$ ($y\in\mathbb{Y}$) and outputs $a\in\mathbb{A}$ ($b\in\mathbb{B}$) are labelled by elements of finite sets. For each choice of inputs $x$ and $y$, the outputs may be produced probabilistically, thus  there is the (conditional) probability distribution $p(a \in \mathbb{A},b \in \mathbb{B}|x \in \mathbb{X},y \in \mathbb{Y})$, which we denote $p(ab|xy)$, for brevity's sake. Charlie's quantum system is associated with a Hilbert space $\mathcal{H}_{C}$ of dimension $D$ and they can perform tomography to characterise the density operator $\rho_C$ for their system. Furthermore, we can consider the sub-normalised density operators for Charlie's system, conditioned upon the classical outputs of the other two parties. That is, if for inputs $x$  and $y$, we obtain outputs $a$ and $b$, there is a positive semi-definite operator $\sigma_{ab|xy}$ acting on $\mathcal{H}_{C}$, such that $\textrm{tr}\{\sigma_{ab|xy}\}=p(ab|xy)$ and $\sum_{a,b}\sigma_{ab|xy}=\rho_C$. 
When taken together, the set of these sub-normalised density operators is called an \textit{assemblage} \cite{pusey2013negativity}, denoted by $\mathbb{\Sigma}:=\{\sigma_{ab|xy}\}_{a,b,x,y}$. Since Alice, Bob, and Charlie are spatially separated, we assume there is no communication between them and thus their shared resources, whatever they may be, cannot be used to communicate. This immediately imposes the constraint that the distribution $p(ab|xy)$ is non-signalling, i.e. $\sum_a p(ab|xy)=p(b|xy)=p(b|y)$ and $\sum_b p(ab|xy)=p(a|xy)=p(a|x)$. These constraints will be reflected in the assemblage, e.g. $\sum_a \sigma_{ab|xy}:=\sigma_{b|y}$ such that $\textrm{tr}\{\sigma_{b|y}\}=p(b|y)$, and the operator $\sigma_{a|x}$ is defined likewise.

A quantum realisation of an assemblage $\mathbb{\Sigma}$ is given by a tripartite quantum state $\ket{\psi}_{ABC}\in\mathcal{H}_{A}\otimes\mathcal{H}_{B}\otimes\mathcal{H}_{C}$ where $\mathcal{H}_{A}$ and $\mathcal{H}_{B}$ are Hilbert spaces for Alice and Bob's systems respectively and  general measurements $\{M_{a|x}\}_{a,x}$ and $\{M_{b|y}\}_{b,y}$ acting on Alice's and Bob's systems respectively. The assemblage elements are:
\begin{align*}
\sigma_{ab|xy}=\textrm{tr}_{AB}\left\{(M_{a|x}\otimes M_{b|y}\otimes \mathbb{I}_{C})\ket{\psi}\bra{\psi}_{ABC}\right\}\,.
\end{align*}
Thus a \textit{post-quantum assemblage} is an assemblage for which there does not exist a state and local measurements such that one recovers the assemblage elements as above.  

\begin{figure}[t!]
    \centering
    \begin{tikzpicture}[scale=0.8]
		\node at (-2,1.2) {Alice};
		\shade[draw, thick, ,rounded corners, inner color=white,outer color=gray!50!white] (-1.7,-0.3) rectangle (-2.3,0.3) ;
		\draw[thick, ->] (-2.5,0.5) to [out=180, in=90] (-2,0.3);
		\node at (-2.9,0.5) {$x$};
		\draw[thick, -<] (-2,-0.3) to [out=-90, in=180] (-2.5,-0.5);
		\node at (-3.1,-0.5) {$a$};
		
		\node at (6,1.2) {Charlie};
		\node[draw,shape=circle,fill,scale=.4] at (6,0) {};
		\node at (6,-1) {$\sigma_{ab|xy}$};
		
 		\node at (2,1.2) {Bob};
 		\shade[draw, thick, ,rounded corners, inner color=white,outer color=gray!50!white] (2.3,-0.3) rectangle (1.7,0.3) ;
 		\draw[thick, ->] (1.5,0.5) to [out=180, in=90] (2,0.3);
 		\node at (1.1,0.5) {$y$};
 		\draw[thick, ->] (2,-0.3) to [out=-90, in=0] (1.5,-0.5);
 		\node at (0.9,-0.5) {$b$};		
	    \end{tikzpicture}
    \caption{Depiction of a tripartite EPR steering scenario, involving two sets of uncharacterised devices (Alice and Bob), and one set of characterised devices (Charlie). The measurements made by Alice and Bob `steer' the state of Charlie, which is fully captured by the set of sub-normalised states $\sigma_{ab|xy}$ (referred to as an `assemblage').}
    \label{fig:tradtrip}
\end{figure}
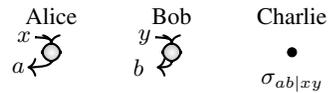

A natural question is how do we show an assemblage is post-quantum. One way to show this is by violating a steering inequality beyond what is possible with quantum assemblages. That is, by construction, we can find a set of Hermitian operators $\mathcal{F}_{abxy}$ (acting on Charlie's Hilbert space) such that the following inequality is satisfied for any quantum assemblage $\mathbb{\Sigma}^{Q}$ with elements $\sigma^{\mathsf{Q}}_{ab|xy}$:
\begin{align}\label{ineq}
\sum_{a\in\mathbb{A}\,,\,b\in\mathbb{B}\,,\,x\in\mathbb{X}\,,\,y\in\mathbb{Y}} \tr[\mathcal{F}_{abxy} \, \sigma^{\mathsf{Q}}_{ab|xy}] \geq 0 \,.
\end{align}
Thus, if an assemblage $\mathbb{\Sigma}$, when put into the left-hand-side of  Eq.~\eqref{ineq},  results in a negative (real) number, then that assemblage is post-quantum. This technique has been used in previous work to demonstrate post-quantum steering \cite{sainz2015postquantum,sainz2018formalism}.

Note that given a set of general measurements $\{M_{c|z}\}_{c,z}$ on Charlie's system, with $z\in\mathbb{Z}$ and $c\in\mathbb{C}$ being the choice of measurement and label of the output respectively, we can obtain the \textit{correlations} $p(abc|xyz)=\textrm{tr}\left(M_{c|z}\sigma_{ab|xy}\right)$, as shown in Fig.~\ref{fig:nopqc}. By construction, these correlations satisfy the tripartite non-signalling conditions, i.e. for all bipartitions of the three parties the non-signalling conditions hold. This recipe for generating correlations allows us to relate the EPR scenario to a Bell scenario, which is fully device-independent. In a Bell test, correlations $p(abc|xyz)$ are said to be post-quantum if there do not exist Hilbert spaces $\mathcal{K}_{A}$, $\mathcal{K}_{B}$ and $\mathcal{K}_{C}$ for the systems of Alice, Bob and Charlie respectively, a quantum state $\ket{\phi}\in\mathcal{K}_{A}\otimes\mathcal{K}_{B}\otimes\mathcal{K}_{C}$, and local measurements $\{N_{a|x}\}_{a,x}$, $\{N_{b|y}\}_{b,y}$ and $\{N_{c|z}\}_{c,z}$ on the systems of Alice, Bob and Charlie respectively such that $p(abc|xyz)=\bra{\phi}N_{a|x}\otimes N_{b|y}\otimes N_{c|z}\ket{\phi}$. Equivalently, post-quantum correlations can be demonstrated by violating a Tsirelson-type bound for a Bell inequality \cite{tsirel1987quantum,PhysRevA.73.022110}.  It should now be clear that if correlations obtained from an assemblage are  post-quantum, then the original assemblage must have been post-quantum. What  Ref.~\cite{sainz2015postquantum} showed was that there exist post-quantum assemblages that do not result in  such post-quantum correlations.

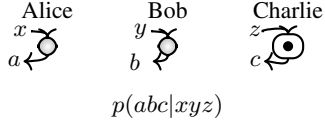
\begin{figure}
    \centering
    \begin{tikzpicture}[scale=0.8]
		\node at (-2,1.2) {Alice};
		\shade[draw, thick, ,rounded corners, inner color=white,outer color=gray!50!white] (-1.7,-0.3) rectangle (-2.3,0.3) ;
		\draw[thick, ->] (-2.5,0.5) to [out=180, in=90] (-2,0.3);
		\node at (-2.9,0.5) {$x$};
		\draw[thick, -<] (-2,-0.3) to [out=-90, in=180] (-2.5,-0.5);
		\node at (-3.1,-0.5) {$a$};
		
		\node at (6,1.2) {Charlie};
		\node[draw,shape=circle,fill,scale=.4] at (6,0) {};
%		\node at (6,-1) {$\sigma_{ab|xy}$};
 		\draw[thick, ,rounded corners] (6.5,-0.4) rectangle (5.5,0.4) ;
   \draw[thick, ->] (5.2,0.5) to [out=180, in=90] (6,0.4);
 		\node at (4.9,0.5) {$z$};
 		\draw[thick, ->] (6,-0.4) to [out=-90, in=0] (5.2,-0.5);
 		\node at (4.9,-0.5) {$c$};		
		
 		\node at (2,1.2) {Bob};
 		\shade[draw, thick, ,rounded corners, inner color=white,outer color=gray!50!white] (2.3,-0.3) rectangle (1.7,0.3) ;
 		\draw[thick, ->] (1.5,0.5) to [out=180, in=90] (2,0.3);
 		\node at (1.1,0.5) {$y$};
 		\draw[thick, ->] (2,-0.3) to [out=-90, in=0] (1.5,-0.5);
 		\node at (0.9,-0.5) {$b$};		
            \node at (2,-2) {$p(abc|xyz)$};
    \end{tikzpicture}
    \caption{If we assume that Charlie performs measurements labelled by $z$, with outcomes labelled by $c$, then an assemblage $\sigma_{ab|xy}$ leads to correlations $p(abc|xyz)$.} 
    \label{fig:nopqc}
\end{figure}

However, and as pointed out in  Ref.~\cite{schmid2021postquantum}, there may be other ways of mapping a post-quantum assemblage to correlations than just by measuring Charlie's system. In this work we describe another method of transforming a post-quantum assemblage into correlations. The method consists of considering a tripartite assemblage distributed between a network of four, non-communicating parties.  The fourth party is called Dani and is treated as a black box, with input $w\in\mathbb{W}$ and output $d\in\mathbb{D}$.  As before, Alice, Bob, and Charlie share the system described by the original tripartite assemblage, but Charlie and Dani also share a bipartite resource. Charlie's Hilbert space is now $\mathcal{H}_{C}\otimes\mathcal{H}_{C'}$, where $\mathcal{H}_{C}$ ($\mathcal{H}_{C'}$) is the space for the resource shared with Alice and Bob (Dani). Therefore, if $\mathbb{\Sigma}_{AB}:=\{\sigma_{ab|xy}\}_{a,b,x,y}$ denotes the assemblage describing Alice, Bob and Charlie's shared resources, and the corresponding assemblage $\mathbb{\Sigma}_{D}:=\{\sigma_{d|w}\}_{d,w}$ for Charlie and Dani, then Charlie's complete system is described by the assemblage $\mathbb{\Sigma}_{AB}\otimes\mathbb{\Sigma}_{D}:=\{\sigma_{ab|xy}\otimes\sigma_{d|w}\}_{a,b,c,d,w,x,y,z}$. 
The four parties can now map the assemblage $\mathbb{\Sigma}_{AB}\otimes\mathbb{\Sigma}_{D}$  into correlations $p(abcd|xyzw)$ by measuring Charlie's system, as in Fig.~\ref{fig:bowlset}. We can now state the main conclusion of our work:

\bigskip

\noindent
\textbf{Main result:} \textit{A tripartite assemblage $\mathbb{\Sigma}_{AB}$ is post-quantum if and only if it is possible to generate correlations $p(abcd|xyzw)$ that are post-quantum for the four-party network scenario depicted in Fig.~\ref{fig:bowlset}.}

\bigskip

What this says is that it is always possible to `activate' post-quantum steering, so that it leads to post-quantum correlations, by placing it in a simple network.

The first thing to notice is that the only post-quantum resource in the assemblage $\mathbb{\Sigma}_{AB}\otimes\mathbb{\Sigma}_{D}$ will be from $\mathbb{\Sigma}_{AB}$.  Indeed, due to the GHJW theorem, $\mathbb{\Sigma}_{D}$ will always have a quantum realisation since it is a bipartite assemblage. Secondly, it is clear that if the correlations $p(abcd|xyzw)$ are post-quantum then the original assemblage $\mathbb{\Sigma}_{AB}$ was post-quantum.  Hence we need to prove the implication in the other direction.  In what remains of the main text and appendix, we show how to prove this when the Hilbert space $\mathcal{H}_{C}$ is dimension $2$, i.e. a qubit. The general case  for arbitrary $d$  is entirely proven  in the appendix (Appendix \ref{ap:highdim}).  

In general, our proof strategy starts with us assuming an arbitrary assemblage $\mathbb{\Sigma}:=\{\sigma_{abd|xyw}\}_{a,b,d,x,y,w}$ for Charlie in the four-party setting, and is based on the correlations $p(abcd|xyzw)$. We can certify, or self-test, that the assemblage $\mathbb{\Sigma}$ is  essentially  of the form $\mathbb{\Sigma}_{AB}\otimes\mathbb{\Sigma}_{D}$, with $\mathbb{\Sigma}_{D}$ known. We can then leverage this to perform particular measurements on Charlie's combined system that result in post-quantum correlations. To explain in more detail, we first need to introduce more formally self-testing of assemblages.

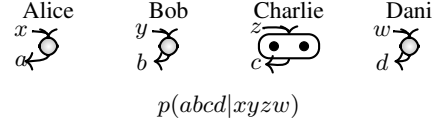
\begin{figure}
    \centering
    \begin{tikzpicture}[scale=0.8]
		\node at (-2,1.2) {Alice};
		\shade[draw, thick, ,rounded corners, inner color=white,outer color=gray!50!white] (-1.7,-0.3) rectangle (-2.3,0.3) ;
		\draw[thick, ->] (-2.5,0.5) to [out=180, in=90] (-2,0.3);
		\node at (-2.9,0.5) {$x$};
		\draw[thick, -<] (-2,-0.3) to [out=-90, in=180] (-2.5,-0.5);
		\node at (-2.9,-0.5) {$a$};
		
		\node at (6,1.2) {Charlie};
		\node[draw,shape=circle,fill,scale=.4] at (5.5,0) {};
            \node[draw,shape=circle,fill,scale=.4] at (6.5,0) {};
%		\node at (6,-1) {$\sigma_{ab|xy}$};
 		\draw[thick, ,rounded corners] (7,-0.4) rectangle (5,0.4) ;
   \draw[thick, ->] (5.2,0.6) to [out=180, in=90] (6,0.4);
 		\node at (4.9,0.6) {$z$};
 		\draw[thick, ->] (6,-0.4) to [out=-90, in=0] (5.2,-0.6);
 		\node at (4.9,-0.6) {$c$};

 		\node at (2,1.2) {Bob};
 		\shade[draw, thick, ,rounded corners, inner color=white,outer color=gray!50!white] (2.3,-0.3) rectangle (1.7,0.3) ;
 		\draw[thick, ->] (1.5,0.5) to [out=180, in=90] (2,0.3);
 		\node at (1.1,0.5) {$y$};
 		\draw[thick, ->] (2,-0.3) to [out=-90, in=0] (1.5,-0.5);
 		\node at (1.1,-0.5) {$b$};

		\node at (10,1.2) {Dani};
 		\shade[draw, thick, ,rounded corners, inner color=white,outer color=gray!50!white] (10.3,-0.3) rectangle (9.7,0.3) ;
 		\draw[thick, ->] (9.5,0.5) to [out=180, in=90] (10,0.3);
 		\node at (9.1,0.5) {$w$};
 		\draw[thick, ->] (10,-0.3) to [out=-90, in=0] (9.5,-0.5);
 		\node at (9.1,-0.5) {$d$};

\node at (4,-2) {$p(abcd|xyzw)$};

	    \end{tikzpicture}
    \caption{The novel network scenario considered in this work. We consider a second assemblage shared between Charlie and Dani, and allow Charlie to perform joint measurements on both of his systems, producing correlations $p(abcd|xyzw)$. Our main result is to show that every post-quantum assemblage will lead to post-quantum correlations in this network configuration.}
    \label{fig:bowlset}
\end{figure}

\section{Self-testing assemblages}

The self-testing of assemblages was introduced in Ref.~\cite{chen2021robust}. Here, we just need to consider the bipartite setting of Charlie and Dani. In traditional self-testing of assemblages we assume the two parties share a quantum system and make local measurements on it to generate correlations $p(cd|zw)$. Since we assume that the systems are quantum, we can consider the ``underlying" assemblage $\mathbb{\Omega}_{D}$ generated for Charlie's system, which has elements  $\sigma_{d|w}=\textrm{tr}_{D}\left\{(\mathbb{I}_{C}\otimes M_{d|w})\rho_{CD}\right\}$.  There, $\rho_{CD}$ is the state shared by Charlie and Dani, and $\{M_{d|w}\}_{d,w}$ are the measurements made by Dani, where $\mathcal{H}_D$ is the Hilbert space of Dani's system.  Note that we do not specify the dimensions of the Hilbert spaces upon which $\rho_{CD}$ acts, and thus $\sigma_{d|w}$ is an operator on some unknown Hilbert space. We want to be able to relate the underlying assemblage $\mathbb{\Omega}_{D}$ to a \textit{reference} assemblage $\mathbb{\Sigma}_{D}$ by some appropriate transformation. That is, we say correlations $p(cd|zw)$ self-test a reference assemblage if there exists a channel $\Lambda$ and density operator $\xi$ such that for all elements $\sigma_{d|w}$ and $\sigma^{\mathsf{R}}_{d|w}$  (in $\mathbb{\Omega}_{D}$ and $\mathbb{\Sigma}_{D}$ respectively)  we have $\xi\otimes\sigma^{\mathsf{R}}_{d|w}=\Lambda(\sigma_{d|w})$. Thus, the assemblage shared by Charlie and Dani is essentially the reference assemblage along with the state of some arbitrary auxiliary system.

We need to self-test an assemblage with elements $\sigma_{d|w}^{\mathsf{R}}$ where $w\in\{1,2,3\}=:\mathbb{W}$, and $d\in\{0,1\}=:\mathbb{D}$ given by
\begin{align}\label{eq:theSTdass}
\sigma^{\mathsf{R}}_{d|w} = r \, \widetilde{\sigma}^{\mathsf{R}}_{d|w}\otimes |0\rangle\langle 0| + (1-r) \, (\widetilde{\sigma}^{\mathsf{R}}_{d|w})^\mathsf{T}\otimes |1\rangle\langle 1|\,,
\end{align}
where $0 \leq r \leq 1$ is an unknown real parameter, $(\cdot)^\mathsf{T}$ denotes the transpose, and  
\begin{align}\label{eq:theassembD}
\widetilde{\sigma}^{\mathsf{R}}_{d|w} = \frac{\mathbb{I} + (-1)^d \, P_w}{2},
\end{align}
where $P_{w}\in\{X,Y,Z\}$ is a Pauli operator depending on $w$.

For the correlations $p(cd|zw)$ that self-test the assemblage $\sigma_{d|w}^{\mathsf{R}}$, we have $c\in\{0,1\}=:\mathbb{C}$ and $z\in\{1,2,3,4,5,6\}$.  However, as is common in self-testing, instead of computing the correlations directly, we can consider a particular Bell inequality violation. The relevant inequality is
\begin{multline}\label{bellineq}
    \mathcal{I}^{\mathsf{CD}} := E_{1,1} + E_{1,2} - E_{2,1} + E_{2,2} + E_{3,1} + E_{4,1} \\
     + E_{3,3} - E_{4,3} + E_{5,2} + E_{6,2} + E_{5,3} - E_{6,3}\leq 6,
\end{multline}
where $E_{z,w} = \sum_{c,d} (-1)^{c+d} \, p(cd|zw)$. It has been shown in Refs.~\cite{bowles2018self,Chen2021robustselftestingof} that for the maximal violation $\mathcal{I}^{\mathsf{CD}}=6\sqrt{2}$, the assemblage $\sigma_{d|w}^{\mathsf{R}}$ above is self-tested. At this point, for simplicity, we assume $r=0$, then each reference assemblage element is $(\sigma_{d|w}^{\mathsf{R}})^{\mathsf{T}}=\tilde{\sigma}_{d|w}^{\mathsf{R}}\otimes |1\rangle\langle 1|$. This also allows us to also omit the ``flag" state $|1\rangle\langle 1|$ from the following discussion. 

This self-testing result applies in the bipartite setting, but our scenario is that of Fig.~\ref{fig:bowlset}, where the resources between the four parties could be post-quantum. As explained fully in   Appendix \ref{ap:apA},  this does not alter the self-testing statement. That is, given correlations $p(abcd|xyzw)$, if they are substituted into the Bell expression $\mathcal{I}^{\mathsf{CD}}$ (by taking the appropriate marginal distributions) and give the maximal violation, then we still self-test the above assemblage $\sigma_{d|w}^{\mathsf{R}}$ on Charlie's system. This is done by leveraging the fact that Charlie's system must be quantum mechanical, and thus the GHJW theorem tells us that all assemblages based on bipartite correlations (between Charlie and Dani) have a quantum realisation. 

A central consequence of self-testing is that if Dani makes a measurement $M_{d|w}$, then Charlie's (subnormalised) post-measurement state can be assumed to be $\xi\otimes(\widetilde{\sigma}_{d|w}^{\mathsf{R}})^{\mathsf{T}}$. For simplicity, assume that $\xi$ is a qubit. Now, suppose that Charlie makes a two-outcome measurement $\{M_{0},M_1=\mathbb{I}-M_{0}\}$ where $M_0=\ket{\Psi^{+}}\bra{\Psi^{+}}$ for $\ket{\Psi^{+}}=\frac{1}{\sqrt{2}}\left(\ket{00}+\ket{11}\right)$. For outcome $0$, we see that 
\begin{equation}\label{useful}
    \textrm{tr}\left(M_{0}(\xi\otimes\sigma_{d|w}^{\mathsf{R}})\right)=\frac{1}{2}\textrm{tr}\left((\sigma_{d|w}^{\mathsf{R}})^{\mathsf{T}}\xi\right).
\end{equation}
Thus, up to a positive real scalar, the entangling measurement allows us to compute the product between  the  known operator $(\sigma_{d|w}^{\mathsf{R}})^{\mathsf{T}}$ and an unknown one.  However, after Alice and Bob generate their outputs $a$ and $b$ (given $x$ and $y$),  Charlie's system is described by assemblage elements $\sigma_{abd|xyw}$, and Eq.~\eqref{useful} does not immediately apply. Fortunately, due to self-testing, in  Appendix \ref{ap:indasump}  we show that all elements  $\sigma_{abd|xyw}$ are of the form $\sigma_{ab|xy}\otimes\sigma_{d|w}^{\mathsf{R}}$. Thus the same approach as in Eq.~\eqref{useful} applies (when $r=0$ as mentioned).

 For each value of $w$,  $\tilde{\sigma}_{0|w}^{\mathsf{R}}-\tilde{\sigma}_{1|w}^{\mathsf{R}}= P_{w}$.  Along with the identity operator, the Pauli operators form an orthonormal basis in the space of two-dimensional Hermitian operators. Additionally, note that for all $w$,  $\tilde{\sigma}_{0|w}^{\mathsf{R}}+\tilde{\sigma}_{1|w}^{\mathsf{R}}=\mathbb{I}$.   Thus every two-dimensional Hermitian operator $F$ can be written as a real linear combination $F=\sum_{d,w}f_{dw}\tilde{\sigma}_{d|w}^{\mathsf{R}}$ for appropriately chosen real coefficients $f_{dw}$. This observation will be central to demonstrating post-quantum correlations, since it allows us to effectively transform steering inequalities into Bell inequalities. 

\section{Tsirelson-Bell inequalities from steering inequalities}

Using the above method, in the steering inequality in Eq.~\eqref{ineq}, the Hermitian operators $\mathcal{F}_{abxy}$ can be replaced by the elements of the reference assemblage, i.e. 
\begin{align}\label{eq:lala}
 \sum_{a,b,d,x,y,w} f_{abdxyw}\tr[\tilde{\sigma}_{d|w}^{\mathsf{R}} \, \sigma^{\mathsf{Q}}_{ab|xy}] \geq 0 \,,
\end{align}
for appropriately chosen coefficients $f_{abdxyw}$. Now, in Sec.~\ref{sse:tsiineq} and Appendix \ref{ap:apB}  we show that if we let Charlie perform  any  joint measurement $\{M_{0},M_1=\mathbb{I}-M_{0}\}$  on their systems, denoted by the setting  $z = \star$, then
\begin{align}\label{pseudobell}
 \sum_{a,b,d,x,y,w} f_{abdxyw}p(ab0d|xy\star w) \geq 0 
\end{align}
(with the same coefficients $f_{abdxyw}$) is bona fide a Tsirelson inequality: any violation of it indicates that the correlations are post-quantum.  Additionally, all of the simplifying assumptions made above can be generalised into a single assumption that the systems shared between Alice, Bob and Charlie are independent of those shared between Charlie and Dani. In Appendix \ref{ap:indasump} we show how the independence assumption follows from self-testing in the special case where $r\in\{0,1\}$ in Eq.~\eqref{eq:theSTdass}. We leave it for future work how to eliminate this assumption.

\subsection{Towards a proof of Eq.~\eqref{pseudobell}}\label{sse:tsiineq}

Recall that Eq.~\eqref{eq:lala} holds for all quantum assemblages with elements $\sigma^{\mathsf{Q}}_{ab|xy}$. Using Eq.~\eqref{useful} and the linearity of the trace, we can rewrite this inequality as
\begin{align}\label{inter1}
 \sum_{a,b,d,x,y,w} 2f_{abdxyw}\tr\{M_{0}((\tilde{\sigma}_{d|w}^{\mathsf{R}})^{\mathsf{T}} \otimes\sigma^{\mathsf{Q}}_{ab|xy})\} \geq 0 \,.
\end{align}
Note that we can remove the factor of $2$ (and the inequality still holds) and recall that a global transpose does not affect the trace. Therefore
\begin{align} \nonumber
 \sum_{a,b,d,x,y,w} f_{abdxyw}\tr\{M_{0}^{\mathsf{T}}(\tilde{\sigma}_{d|w}^{\mathsf{R}} \otimes(\sigma^{\mathsf{Q}}_{ab|xy})^{\mathsf{T}})\}
= \\ \sum_{a,b,d,x,y,w} f_{abdxyw}\tr\{M_{0}(\tilde{\sigma}_{d|w}^{\mathsf{R}} \otimes\tilde{\sigma}^{\mathsf{Q}}_{ab|xy})\} \geq 0, \label{inter2}
\end{align}
where in the second line we used that fact that the maximally entangled state is invariant under transpose and that, as shown in Ref.~\cite{sainz2018formalism}, the transpose of a quantum assemblage gives a new quantum assemblage (denoted by the tilde). Since the inequality holds for all quantum assemblages, we still get positivity. By this same argument, we now see that the positivity in Eq.~\eqref{inter2} is satisfied even if we replace the operator $\tilde{\sigma}_{d|w}^{\mathsf{R}}$ in  Eq.~\eqref{inter2}  with $\sigma^{\mathsf{R}}_{d|w} = r \, \widetilde{\sigma}^{\mathsf{R}}_{d|w}\otimes|0\rangle\langle0| + (1-r) \, (\widetilde{\sigma}^{\mathsf{R}}_{d|w})^\mathsf{T}\otimes|1\rangle\langle 1|$ for any value of $0\leq r \leq 1$.  Since we have an additional qubit storing the ``flag'' states $|0\rangle\langle 0|$ and $|1\rangle\langle 1|$, the measurement operator $M_{0}$ is extended to $M_{0}\otimes \mathbb{I}$, with the identity operator acting on the flag system.

If a measurement $\{M_{0},M_1=\mathbb{I}-M_{0}\}$ is made on Charlie's system, we will denote this by the measurement choice $\star$, and the outcomes with $c\in\{0,1\}$. In our scenario Charlie's final set of measurement choices is then $\mathbb{Z}:=\{1,2,3,4,5,6,\star\}$, with the values of $\mathbb{C}:=\{0,1\}$. Thus we have $\tr\{M_{0}(\sigma_{d|w}^{\mathsf{R}} \otimes\sigma^{\mathsf{Q}}_{ab|xy})\}=p(ab0d|xy\star w)$. By substituting this probability into  Eq.~\eqref{inter2}  we arrive at Eq.~\eqref{pseudobell}
This resembles a Bell inequality, since it is a linear combination of probabilities. A seeming issue is that it might not hold for all quantum correlations $p(ab0d|xy\star w)$. That is, in deriving this inequality we have made the following assumptions: (i) part of Charlie's system is a qubit described by assemblage elements $\sigma_{d|w}^{\mathsf{R}}$; (ii) Charlie has an additional qubit which is described by assemblage elements $\sigma^{\mathsf{Q}}_{ab|xy}$; (iii) each of Charlie's assemblage elements $\sigma_{abd|xyw}$ is equal to $\sigma_{ab|xy}\otimes\sigma_{d|w}^{\mathsf{R}}$, and (iv) Charlie makes the measurement associated with the choice $\star$. The first assumption can be \textit{concluded} by self-testing the assemblage shared by Charlie and Dani alone, and is therefore in fact not an assumption. In Appendix \ref{ap:apB}, we show that we can dispense with the second and fourth assumptions. That is, there we show that assuming the self-testing was successful and assumption (iii) above, the inequality in  Eq.~\eqref{pseudobell} is a Tsirelson inequality: any violation of it indicates that the correlations are post-quantum.

\section{Summary and discussions} 

To summarise all of the above, we have shown that if a tripartite post-quantum assemblage $\mathbb{\Sigma}$ is shared between Alice, Bob, and Charlie, we can distribute this in a network of four parties (including Dani). Then we can consider the following Bell scenario: Charlie's set of measurement choices in $\mathbb{Z}=\{1,2,3,4,5,6,\star\}$ and Dani's set is $\mathbb{W}:=\{1,2,3\}$ and their outcomes are $\mathbb{C}=\mathbb{D}=\{0,1\}$, with Alice and Bob's measurements and outcomes determined by the assemblage. If Charlie and Dani's correlations maximally violate the inequality in  Eq.~\eqref{bellineq} by sharing the quantum assemblage in  Eq.~\eqref{eq:theSTdass}, then there exists an inequality of the form of  Eq.~\eqref{pseudobell} which they violate and thus demonstrate post-quantum correlations. The above is true even if the tripartite post-quantum assemblage does not demonstrate post-quantum correlations in the setting in Fig.~\ref{fig:nopqc} (i.e.~in the setting considered in the original work \cite{sainz2015postquantum}). It is thus always possible, using this procedure, to `activate' post-quantum steering, so that we see the post-quantumness already at the level of correlations alone. 

For higher (but finite) dimensional systems, we can generalise the construction above to make the result universal  (see Appendix \ref{ap:highdim}).   This activation technique  can also be straightforwardly generalised to the case of more than three parties in the initial steering scenario.  For other generalised forms of EPR steering, activation of the relevant kind of generalised post-quantum steering have also been found \cite{zjawin2024activation}.
It would be interesting to derive both our argument and the one in Ref.~\cite{zjawin2024activation} into a single expression or inequality that witnesses post-quantum non-locality from a post-quantum assemblage; we leave this to future work.

There is a resource-theoretic aspect to our work. In  Ref.~\cite{schmid2021postquantum}, the resource theory of local operations and shared entanglement (LOSE) was explored as a means to understand post-quantumness. The free resources in this resource theory are arbitrary quantum states and perform local operations, but multiple parties cannot communicate.  In some sense, our work can be seen as a resource conversion task where a post-quantum resource is combined with a free resource (the assemblage in  Eq.~\eqref{eq:theSTdass})  to produce a new resource that exhibits post-quantum correlations. This draws a nice parallel with other resource conversion tasks in resource theories such as local operations and classical communication (LOCC).

One interesting aspect of post-quantum steering is that it allows for another point of view on the question of how to single out quantum theory from the landscape of possible physical theories. More specifically, what principles are sufficient (and meaningful) to single out the set of quantum assemblages from the set of all possible assemblages. Our activation result provides a new way to tackle this question, by connecting it to that of quantum correlations and research in device-independent principles to recover them \cite{van2013implausible,brassard2006limit,navascues2010glance,pawlowski2009information,fritz2013local}. The quest to find such a complete set of postulates for correlations has barely progressed in the recent years, and it has even been conjectured that device-dependent postulates will be needed to succeed at the task. The EPR scenario is a bridge between device-independent and device-dependent settings, and offers new ground for  exploration \cite{rossi2022characterising}.  Going further, one could ask in which general physical theories can one have activation of steering in the same manner as our work. Could activation itself result in useful postulates to single out quantum theory? We hope our work adds to the toolbox required to understand what is special about quantum theory.

\bigskip

\section*{Acknowledgments} 

MJH and ABS thank John H. Selby for fruitful discussions.  MJH acknowledges the FQXi large grant ``The Emergence of Agents from Causal Order'' (FQXi FFF Grant number FQXi-RFP-1803B). PS is a CIFAR Azrieli Global Scholar in the Quantum Information Science Program, and also gratefully acknowledges support from a Royal Society University Research Fellowship (UHQT/NFQI). ABS acknowledges support by the Foundation for Polish Science (IRAP project, ICTQT, contract no. 2018/MAB/5, co-financed by EU within Smart Growth Operational Programme), and by the IRA Programme, project no. FENG.02.01-IP.05-0006/23, financed by the FENG program 2021-2027, Priority FENG.02, Measure FENG.02.01., with the support of the FNP.
\newline

\bibliographystyle{quantum}
\bibliography{bibliography}

% \onecolumngrid
 
\appendix
% \newpage

\section{Self-testing in post-quantum scenarios}\label{ap:apA}

In this section we explain how one can self-test a quantum assemblage between Charlie and Dani when the underlying scenario that involves the four parties may allow for post-quantum systems to be shared. 

Let us assume by contradiction that self-testing cannot be performed between Charlie and Dani when the underlying physical theory cannot be guaranteed to be quantum. That is, let Charlie and Dani observe correlations $p(cd|zw)$ that are compatible with the assemblage $\{\sigma_{d|w}\}$ being the quantum assemblage that arises from performing the Pauli measurements on the maximally entangled state (denoted $\{\widetilde{\sigma}^{\mathsf{R}}_{d|w}\}$ in the main text). Our assumption that `self-testing is not possible' then implies that there is another assemblage $\{\sigma^*_{d|w}\}$ that is different from $\{\widetilde{\sigma}^{\mathsf{R}}_{d|w}\}$ and that yields the same correlations $p(cd|zw)$. Now, the GHJW theorem implies that there is a quantum realisation of $\{\sigma^*_{d|w}\}$, which in turns provides a quantum assemblage different (and not equivalent to) $\{\widetilde{\sigma}^{\mathsf{R}}_{d|w}\}$ which yields the same correlations $p(cd|zw)$. This then contradicts the fact that the assemblage $\{\widetilde{\sigma}^{\mathsf{R}}_{d|w}\}$ can be self-tested within quantum theory. It follows then that such a different $\{\sigma^*_{d|w}\}$ cannot exist, and hence self-testing can be performed between Charlie and Dani when the underlying physical theory cannot be guaranteed to be quantum.

\section{Tsirelson-Bell inequalities from steering inequalities -- getting rid of assumptions (ii) and (iv)}\label{ap:apB}

In Sec.~\ref{sse:tsiineq} we begun the proof of the claim that Eq.~\eqref{pseudobell} holds for any correlation that comes from Charlie sharing a quantum assemblage with Alice and Bob. There we made four assumptions: (i) part of Charlie's system is a qubit described by assemblage elements $\sigma_{d|w}^{\mathsf{R}}$; (ii) Charlie has an additional qubit which is described by assemblage elements $\sigma^{\mathsf{Q}}_{ab|xy}$; (iii) each of Charlie's assemblage elements $\sigma_{abd|xyw}$ is equal to $\sigma_{ab|xy}\otimes\sigma_{d|w}^{\mathsf{R}}$, and (iv) Charlie makes the measurement associated with the choice $\star$. We argued that assumption (i) can actually be certified by the self-testing step being successful. Here we will show that we can dispense with the second and fourth assumptions.

Basically, we need  to compute the minimum quantum value of the Bell functional of Eq.~\eqref{pseudobell} under the assumption that the self-testing stage of the protocol is successful and assumption (iii) above. From a technical point of view, this means we will consider the particular case of system $\cH_{C'}$ being prepared on the assemblage with elements given by Eq.~\eqref{eq:theSTdass}. First we will focus on the case where $r=1$ in Eq.~\eqref{eq:theSTdass}, i.e., the assemblage elements actually take the form  $\{\widetilde{\sigma}^{\mathsf{R}}_{d|w}\}$ specified in Eq.~\eqref{eq:theassembD}. From the discussion above, one can then conclude that the result also applies to other values of $r$. 

If we take $M_0 = \ket{\Psi^+}\bra{\Psi^+}$, we have already shown that Eq.~\eqref{pseudobell} holds under the assumption of successful self-testing. What remains to be proven is that the same still holds when we optimise the value of the left-hand side of Eq.~\eqref{pseudobell} over any possible value of the measurement operator $M_0$. The proof of this will be presented using the diagrammatic language of Refs.~\cite{coecke_kissinger_2017,kissinger2018paper}. 

First, note that $\mathcal{F}_{abxy} = \sum_{dw} f_{abdxyw} \, \widetilde{\sigma}^{\mathsf{R}}_{d|w}$, and hence the claim for $M_0 = \ket{\Psi^+}\bra{\Psi^+}$ can be diagrammatically expressed as: 
\begin{align}\label{eq:diag1}
&\sum_{a,b,d,x,y,w} f_{abdxyw}p(ab0d|xy\star w) = \frac{1}{2} \sum_{a,b,x,y}\,%
\begin{tikzpicture}
	\begin{pgfonlayer}{nodelayer}
		\node [style=none] (0) at (-2.5, -0.75) {};
		\node [style=none] (1) at (2.5, -0.75) {};
		\node [style=none] (2) at (0, -1.75) {};
		\node [style=none] (4) at (0, -1.25) {$\rho$};
		\node [style=copoint] (5) at (0, 1.25) {$b|y$};
		\node [style=copoint] (6) at (-2, 1.25) {$a|x$};
		\node [style=none] (7) at (2, 0.75) {};
		\node [style=none] (8) at (-2, -0.75) {};
		\node [style=none] (9) at (0, -0.75) {};
		\node [style=none] (10) at (2, -0.75) {};
		\node [style=none] (12) at (1.25, 0.25) {$\cH_C$};
		\node [style=none] (13) at (-2.75, 0.25) {$\cH_A$};
		\node [style=none] (14) at (-0.75, 0.25) {$\cH_B$};
		\node [style=none] (15) at (3, -0.75) {};
		\node [style=none] (16) at (7, -0.75) {};
		\node [style=none] (17) at (5, -2.25) {};
		\node [style=none] (18) at (5, -1.25) {\color{white}$\mathcal{F}_{abxy}$};
		\node [style=none] (19) at (5, 0.75) {};
		\node [style=none] (20) at (5, -0.75) {};
		\node [style=none] (21) at (6, 0.25) {$\cH_{C'}$};
	\end{pgfonlayer}
	\begin{pgfonlayer}{edgelayer}
		\draw (0.center) to (2.center);
		\draw (2.center) to (1.center);
		\draw (1.center) to (0.center);
		\draw [qWire] (8.center) to (6);
		\draw [qWire] (9.center) to (5);
		\draw [qWire] (10.center) to (7.center);
		\draw [fill=black] (15.center)
			 to (17.center)
			 to (16.center)
			 to cycle;
		\draw [qWire] (20.center) to (19.center);
		\draw [qWire, bend left=90] (7.center) to (19.center);
	\end{pgfonlayer}
\end{tikzpicture}} \, \geq 0 \,,
\end{align}
for all quantum systems $\cH_A$ and $\cH_B$, effects $%
\begin{tikzpicture}
	\begin{pgfonlayer}{nodelayer}
		\node [style=none] (0) at (0, 0.5) {};
		\node [style=copoint] (1) at (0, 0.5) {$a|x$};
		\node [style=none] (2) at (0, -1) {};
	\end{pgfonlayer}
	\begin{pgfonlayer}{edgelayer}
		\draw [style=qWire] (1) to (2.center);
	\end{pgfonlayer}
\end{tikzpicture}
}$ and $%
\begin{tikzpicture}
	\begin{pgfonlayer}{nodelayer}
		\node [style=none] (0) at (0, 0.5) {};
		\node [style=copoint] (1) at (0, 0.5) {$b|y$};
		\node [style=none] (2) at (0, -1) {};
	\end{pgfonlayer}
	\begin{pgfonlayer}{edgelayer}
		\draw [style=qWire] (1) to (2.center);
	\end{pgfonlayer}
\end{tikzpicture}
}$, and state $%
\begin{tikzpicture}
	\begin{pgfonlayer}{nodelayer}
		\node [style=none] (0) at (-1.5, 0) {};
		\node [style=none] (1) at (1.5, 0) {};
		\node [style=none] (2) at (0, -1) {};
		\node [style=none] (4) at (0, -0.5) {$\rho$};
		\node [style=none] (5) at (0, 1) {};
		\node [style=none] (6) at (-1, 1) {};
		\node [style=none] (7) at (1, 1) {};
		\node [style=none] (8) at (-1, 0) {};
		\node [style=none] (9) at (0, 0) {};
		\node [style=none] (10) at (1, 0) {};
	\end{pgfonlayer}
	\begin{pgfonlayer}{edgelayer}
		\draw (0.center) to (2.center);
		\draw (2.center) to (1.center);
		\draw (1.center) to (0.center);
		\draw [qWire] (8.center) to (6.center);
		\draw [qWire] (9.center) to (5.center);
		\draw [qWire] (10.center) to (7.center);
	\end{pgfonlayer}
\end{tikzpicture}}$. Notice that in Eq.~\eqref{eq:diag1} the diagram
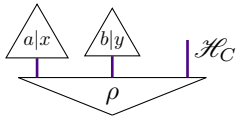
\begin{align*}
\begin{tikzpicture}
	\begin{pgfonlayer}{nodelayer}
		\node [style=none] (0) at (-2.5, 0) {};
		\node [style=none] (1) at (2.5, 0) {};
		\node [style=none] (2) at (0, -1) {};
		\node [style=none] (4) at (0, -0.5) {$\rho$};
		\node [style=copoint] (5) at (0, 1) {$b|y$};
		\node [style=copoint] (6) at (-2, 1) {$a|x$};
		\node [style=none] (7) at (2, 1) {};
		\node [style=none] (8) at (-2, 0) {};
		\node [style=none] (9) at (0, 0) {};
		\node [style=none] (10) at (2, 0) {};
		\node [style=none] (11) at (2.5, 0.75) {};
		\node [style=none] (12) at (2.75, 0.75) {$\cH_C$};
	\end{pgfonlayer}
	\begin{pgfonlayer}{edgelayer}
		\draw (0.center) to (2.center);
		\draw (2.center) to (1.center);
		\draw (1.center) to (0.center);
		\draw [qWire] (8.center) to (6);
		\draw [qWire] (9.center) to (5);
		\draw [qWire] (10.center) to (7.center);
	\end{pgfonlayer}
\end{tikzpicture}} 
\end{align*}
represents an arbitrary quantum assemblage for the quantum system $\cH_C$ held by Charlie. In addition, in the diagram of Eq.~\eqref{eq:diag1} the operators $\mathcal{F}_{abxy}^\mathsf{T}$ are depicted as black triangles\footnote{In general, possibly-unphysical processes will be depicted with a black background.} since they correspond to Hermitian operators which are not necessarily valid quantum states. 

In the diagrammatic language, then, what we need to show is expressed as: 
\begin{align}\label{eq:diag2}
&\sum_{a,b,d,x,y,w} f_{abdxyw}p(ab0d|xy\star w) = \frac{1}{2} \sum_{a,b,x,y}\,%
\begin{tikzpicture}
	\begin{pgfonlayer}{nodelayer}
		\node [style=none] (0) at (-2.5, -0.75) {};
		\node [style=none] (1) at (2.5, -0.75) {};
		\node [style=none] (2) at (0, -1.75) {};
		\node [style=none] (4) at (0, -1.25) {$\rho$};
		\node [style=copoint] (5) at (0, 1.25) {$b|y$};
		\node [style=copoint] (6) at (-2, 1.25) {$a|x$};
		\node [style=none] (7) at (2, 1) {};
		\node [style=none] (8) at (-2, -0.75) {};
		\node [style=none] (9) at (0, -0.75) {};
		\node [style=none] (10) at (2, -0.75) {};
		\node [style=none] (12) at (1.25, 0.25) {$\cH_C$};
		\node [style=none] (13) at (-2.75, 0.25) {$\cH_A$};
		\node [style=none] (14) at (-0.75, 0.25) {$\cH_B$};
		\node [style=none] (15) at (3, -0.75) {};
		\node [style=none] (16) at (7, -0.75) {};
		\node [style=none] (17) at (5, -2.25) {};
		\node [style=none] (18) at (5, -1.25) {\color{white}$\mathcal{F}_{abxy}$};
		\node [style=none] (19) at (5, 1) {};
		\node [style=none] (20) at (5, -0.75) {};
		\node [style=none] (21) at (6, 0.25) {$\cH_{C'}$};
		\node [style=none] (22) at (1.5, 1) {};
		\node [style=none] (23) at (5.5, 1) {};
		\node [style=none] (24) at (3.5, 2.25) {};
		\node [style=none] (25) at (3.5, 1.5) {};
		\node [style=none] (26) at (3.5, 1.5) {};
		\node [style=none] (27) at (3.5, 1.5) {$M_{0|\star}$};
	\end{pgfonlayer}
	\begin{pgfonlayer}{edgelayer}
		\draw (0.center) to (2.center);
		\draw (2.center) to (1.center);
		\draw (1.center) to (0.center);
		\draw [qWire] (8.center) to (6);
		\draw [qWire] (9.center) to (5);
		\draw [qWire] (10.center) to (7.center);
		\draw [fill=black] (15.center)
			 to (17.center)
			 to (16.center)
			 to cycle;
		\draw [qWire] (20.center) to (19.center);
		\draw (24.center) to (23.center);
		\draw (23.center) to (22.center);
		\draw (22.center) to (24.center);
	\end{pgfonlayer}
\end{tikzpicture}} \, \geq 0 \,,
\end{align}
for all quantum systems $\cH_A$ and $\cH_B$, effects $%
\begin{tikzpicture}
	\begin{pgfonlayer}{nodelayer}
		\node [style=none] (0) at (0, 0.5) {};
		\node [style=copoint] (1) at (0, 0.5) {$a|x$};
		\node [style=none] (2) at (0, -1) {};
	\end{pgfonlayer}
	\begin{pgfonlayer}{edgelayer}
		\draw [style=qWire] (1) to (2.center);
	\end{pgfonlayer}
\end{tikzpicture}
}$, $%
\begin{tikzpicture}
	\begin{pgfonlayer}{nodelayer}
		\node [style=none] (0) at (0, 0.5) {};
		\node [style=copoint] (1) at (0, 0.5) {$b|y$};
		\node [style=none] (2) at (0, -1) {};
	\end{pgfonlayer}
	\begin{pgfonlayer}{edgelayer}
		\draw [style=qWire] (1) to (2.center);
	\end{pgfonlayer}
\end{tikzpicture}
}$, and $%
\begin{tikzpicture}
	\begin{pgfonlayer}{nodelayer}
		\node [style=none] (0) at (0, 0.5) {};
		\node [style=copoint] (1) at (0, 0.5) {$M_{0|\star}$};
		\node [style=none] (2) at (0, -1) {};
	\end{pgfonlayer}
	\begin{pgfonlayer}{edgelayer}
		\draw [style=qWire] (1) to (2.center);
	\end{pgfonlayer}
\end{tikzpicture}
}$, and state $%
\begin{tikzpicture}
	\begin{pgfonlayer}{nodelayer}
		\node [style=none] (0) at (-1.5, 0) {};
		\node [style=none] (1) at (1.5, 0) {};
		\node [style=none] (2) at (0, -1) {};
		\node [style=none] (4) at (0, -0.5) {$\rho$};
		\node [style=none] (5) at (0, 1) {};
		\node [style=none] (6) at (-1, 1) {};
		\node [style=none] (7) at (1, 1) {};
		\node [style=none] (8) at (-1, 0) {};
		\node [style=none] (9) at (0, 0) {};
		\node [style=none] (10) at (1, 0) {};
	\end{pgfonlayer}
	\begin{pgfonlayer}{edgelayer}
		\draw (0.center) to (2.center);
		\draw (2.center) to (1.center);
		\draw (1.center) to (0.center);
		\draw [qWire] (8.center) to (6.center);
		\draw [qWire] (9.center) to (5.center);
		\draw [qWire] (10.center) to (7.center);
	\end{pgfonlayer}
\end{tikzpicture}}$.

What we do now is to manipulate the diagram from Eq.~\eqref{eq:diag2}. The first step is to notice that $\rho'$ defined as
\begin{align}\label{eq:diag3}
\begin{tikzpicture}
	\begin{pgfonlayer}{nodelayer}
		\node [style=none] (0) at (-1.75, -0.75) {};
		\node [style=none] (1) at (3.25, -0.75) {};
		\node [style=none] (2) at (0.75, -1.75) {};
		\node [style=none] (4) at (0.75, -1.25) {$\rho$};
		\node [style=none] (5) at (0.75, 1) {};
		\node [style=none] (6) at (-1.25, 1) {};
		\node [style=none] (7) at (2.75, 1) {};
		\node [style=none] (8) at (-1.25, -0.75) {};
		\node [style=none] (9) at (0.75, -0.75) {};
		\node [style=none] (10) at (2.75, -0.75) {};
		\node [style=none] (12) at (2, 0.25) {$\cH_C$};
		\node [style=none] (13) at (-2, 0.25) {$\cH_A$};
		\node [style=none] (14) at (0, 0.25) {$\cH_B$};
		\node [style=none] (19) at (5.75, 1) {};
		\node [style=none] (20) at (5.75, -0.75) {};
		\node [style=none] (21) at (4.75, 0.25) {$\cH_{C'}$};
		\node [style=none] (22) at (2.25, 1) {};
		\node [style=none] (23) at (6.25, 1) {};
		\node [style=none] (24) at (4.25, 2.25) {};
		\node [style=none] (25) at (4.25, 1.5) {};
		\node [style=none] (26) at (4.25, 1.5) {};
		\node [style=none] (27) at (4.25, 1.5) {$M_{0|\star}$};
		\node [style=none] (28) at (7.75, 1) {};
		\node [style=none] (29) at (7.75, -0.75) {};
		\node [style=none] (30) at (8.75, 0.25) {$\cH_{C}$};
		\node [style=none] (31) at (-9, -0.75) {};
		\node [style=none] (32) at (-4, -0.75) {};
		\node [style=none] (33) at (-6.5, -1.75) {};
		\node [style=none] (34) at (-6.5, -1.25) {$\rho^\prime$};
		\node [style=none] (35) at (-6.5, 1) {};
		\node [style=none] (36) at (-8.5, 1) {};
		\node [style=none] (37) at (-4.5, 1) {};
		\node [style=none] (38) at (-8.5, -0.75) {};
		\node [style=none] (39) at (-6.5, -0.75) {};
		\node [style=none] (40) at (-4.5, -0.75) {};
		\node [style=none] (41) at (-5.25, 0.25) {$\cH_C$};
		\node [style=none] (42) at (-9.25, 0.25) {$\cH_A$};
		\node [style=none] (43) at (-7.25, 0.25) {$\cH_B$};
		\node [style=none] (44) at (-3, 0) {$=$};
	\end{pgfonlayer}
	\begin{pgfonlayer}{edgelayer}
		\draw (0.center) to (2.center);
		\draw (2.center) to (1.center);
		\draw (1.center) to (0.center);
		\draw [qWire] (8.center) to (6.center);
		\draw [qWire] (9.center) to (5.center);
		\draw [qWire] (10.center) to (7.center);
		\draw [qWire] (20.center) to (19.center);
		\draw (24.center) to (23.center);
		\draw (23.center) to (22.center);
		\draw (22.center) to (24.center);
		\draw [qWire] (29.center) to (28.center);
		\draw [qWire, bend right=90, looseness=1.25] (20.center) to (29.center);
		\draw (31.center) to (33.center);
		\draw (33.center) to (32.center);
		\draw (32.center) to (31.center);
		\draw [qWire] (38.center) to (36.center);
		\draw [qWire] (39.center) to (35.center);
		\draw [qWire] (40.center) to (37.center);
	\end{pgfonlayer}
\end{tikzpicture}}
\end{align}
is a valid quantum state, though possibly unnormalised. 

Hence, 
\begin{align}\label{eq:diag4}
&\sum_{a,b,x,y} %
\begin{tikzpicture}
	\begin{pgfonlayer}{nodelayer}
		\node [style=none] (0) at (-2.5, -0.75) {};
		\node [style=none] (1) at (2.5, -0.75) {};
		\node [style=none] (2) at (0, -1.75) {};
		\node [style=none] (4) at (0, -1.25) {$\rho$};
		\node [style=copoint] (5) at (0, 1.25) {$b|y$};
		\node [style=copoint] (6) at (-2, 1.25) {$a|x$};
		\node [style=none] (7) at (2, 1) {};
		\node [style=none] (8) at (-2, -0.75) {};
		\node [style=none] (9) at (0, -0.75) {};
		\node [style=none] (10) at (2, -0.75) {};
		\node [style=none] (12) at (1.25, 0.25) {$\cH_C$};
		\node [style=none] (13) at (-2.75, 0.25) {$\cH_A$};
		\node [style=none] (14) at (-0.75, 0.25) {$\cH_B$};
		\node [style=none] (15) at (3, -0.75) {};
		\node [style=none] (16) at (7, -0.75) {};
		\node [style=none] (17) at (5, -2.25) {};
		\node [style=none] (18) at (5, -1.25) {\color{white}$\mathcal{F}_{abxy}$};
		\node [style=none] (19) at (5, 1) {};
		\node [style=none] (20) at (5, -0.75) {};
		\node [style=none] (21) at (6, 0.25) {$\cH_{C'}$};
		\node [style=none] (22) at (1.5, 1) {};
		\node [style=none] (23) at (5.5, 1) {};
		\node [style=none] (24) at (3.5, 2.25) {};
		\node [style=none] (25) at (3.5, 1.5) {};
		\node [style=none] (26) at (3.5, 1.5) {};
		\node [style=none] (27) at (3.5, 1.5) {$M_{0|\star}$};
	\end{pgfonlayer}
	\begin{pgfonlayer}{edgelayer}
		\draw (0.center) to (2.center);
		\draw (2.center) to (1.center);
		\draw (1.center) to (0.center);
		\draw [qWire] (8.center) to (6);
		\draw [qWire] (9.center) to (5);
		\draw [qWire] (10.center) to (7.center);
		\draw [fill=black] (15.center)
			 to (17.center)
			 to (16.center)
			 to cycle;
		\draw [qWire] (20.center) to (19.center);
		\draw (24.center) to (23.center);
		\draw (23.center) to (22.center);
		\draw (22.center) to (24.center);
	\end{pgfonlayer}
\end{tikzpicture}} = \sum_{a,b,x,y}%
\begin{tikzpicture}
	\begin{pgfonlayer}{nodelayer}
		\node [style=none] (0) at (-2.5, -0.75) {};
		\node [style=none] (1) at (2.5, -0.75) {};
		\node [style=none] (2) at (0, -1.75) {};
		\node [style=none] (4) at (0, -1.25) {$\rho^\prime$};
		\node [style=copoint] (5) at (0, 1.25) {$b|y$};
		\node [style=copoint] (6) at (-2, 1.25) {$a|x$};
		\node [style=none] (7) at (2, 0.75) {};
		\node [style=none] (8) at (-2, -0.75) {};
		\node [style=none] (9) at (0, -0.75) {};
		\node [style=none] (10) at (2, -0.75) {};
		\node [style=none] (12) at (1.25, 0.25) {$\cH_C$};
		\node [style=none] (13) at (-2.75, 0.25) {$\cH_A$};
		\node [style=none] (14) at (-0.75, 0.25) {$\cH_B$};
		\node [style=none] (15) at (3, -0.75) {};
		\node [style=none] (16) at (7, -0.75) {};
		\node [style=none] (17) at (5, -2.25) {};
		\node [style=none] (18) at (5, -1.25) {\color{white}$\mathcal{F}_{abxy}$};
		\node [style=none] (19) at (5, 0.75) {};
		\node [style=none] (20) at (5, -0.75) {};
		\node [style=none] (21) at (6, 0.25) {$\cH_{C'}$};
	\end{pgfonlayer}
	\begin{pgfonlayer}{edgelayer}
		\draw (0.center) to (2.center);
		\draw (2.center) to (1.center);
		\draw (1.center) to (0.center);
		\draw [qWire] (8.center) to (6);
		\draw [qWire] (9.center) to (5);
		\draw [qWire] (10.center) to (7.center);
		\draw [fill=black] (15.center)
			 to (17.center)
			 to (16.center)
			 to cycle;
		\draw [qWire] (20.center) to (19.center);
		\draw [qWire, bend left=90] (7.center) to (19.center);
	\end{pgfonlayer}
\end{tikzpicture}} \geq 0
\end{align}
for all quantum systems $\cH_A$ and $\cH_B$, effects $%
\begin{tikzpicture}
	\begin{pgfonlayer}{nodelayer}
		\node [style=none] (0) at (0, 0.5) {};
		\node [style=copoint] (1) at (0, 0.5) {$a|x$};
		\node [style=none] (2) at (0, -1) {};
	\end{pgfonlayer}
	\begin{pgfonlayer}{edgelayer}
		\draw [style=qWire] (1) to (2.center);
	\end{pgfonlayer}
\end{tikzpicture}
}$, $%
\begin{tikzpicture}
	\begin{pgfonlayer}{nodelayer}
		\node [style=none] (0) at (0, 0.5) {};
		\node [style=copoint] (1) at (0, 0.5) {$b|y$};
		\node [style=none] (2) at (0, -1) {};
	\end{pgfonlayer}
	\begin{pgfonlayer}{edgelayer}
		\draw [style=qWire] (1) to (2.center);
	\end{pgfonlayer}
\end{tikzpicture}
}$, and $%
\begin{tikzpicture}
	\begin{pgfonlayer}{nodelayer}
		\node [style=none] (0) at (0, 0.5) {};
		\node [style=copoint] (1) at (0, 0.5) {$M_{0|\star}$};
		\node [style=none] (2) at (0, -1) {};
	\end{pgfonlayer}
	\begin{pgfonlayer}{edgelayer}
		\draw [style=qWire] (1) to (2.center);
	\end{pgfonlayer}
\end{tikzpicture}
}$, and state $%
\begin{tikzpicture}
	\begin{pgfonlayer}{nodelayer}
		\node [style=none] (0) at (-1.5, 0) {};
		\node [style=none] (1) at (1.5, 0) {};
		\node [style=none] (2) at (0, -1) {};
		\node [style=none] (4) at (0, -0.5) {$\rho$};
		\node [style=none] (5) at (0, 1) {};
		\node [style=none] (6) at (-1, 1) {};
		\node [style=none] (7) at (1, 1) {};
		\node [style=none] (8) at (-1, 0) {};
		\node [style=none] (9) at (0, 0) {};
		\node [style=none] (10) at (1, 0) {};
	\end{pgfonlayer}
	\begin{pgfonlayer}{edgelayer}
		\draw (0.center) to (2.center);
		\draw (2.center) to (1.center);
		\draw (1.center) to (0.center);
		\draw [qWire] (8.center) to (6.center);
		\draw [qWire] (9.center) to (5.center);
		\draw [qWire] (10.center) to (7.center);
	\end{pgfonlayer}
\end{tikzpicture}}$, where the last step (the claim that $\geq 0$) follows from the fact that Eq.~\eqref{eq:diag1} is valid for all quantum states $\rho$. This completes the proof of our claim.

Regarding assumption (iii) above, we can ensure that it holds by assuming that the system Charlie shares with Dani is independent of the system Charlie shares with Alice and Bob. Furthermore, if we assume that $r\in\{0,1\}$ for Eq.~\eqref{eq:theSTdass} after self-testing, then we can show that assumption (iii) follows. We conjecture that assumption (iii) follows in the more general case where $r\in[0,1]$ in Eq.~\eqref{eq:theSTdass}. Proving this is left for future work.

\section{Activation of post-quantum steering for higher-dimensional assemblages}\label{ap:highdim}

In this Appendix we discuss how to extend the methods from the main paper so they apply to assemblages in a Hilbert space of dimension higher than two. We begin by stating the three main points of departure from the qubit protocol, based on the work presented in Ref.~\cite{bowles2018device,bowles2018self}, and afterwards present the qudit protocol.

\bigskip

Consider the setup of Fig.~\ref{fig:bowlset}, and let $\cH_C \equiv \cH_d$, where the dimension is $d > 2$. The first step is to embed $\cH_d$ within a Hilbert space with dimension `power of 2'. That is, we will think of a qudit from $\cH_d$ as living in a Hilbert space that arises from the parallel composition of many qubits. We denote by $n$ the number of qubits, from which follows that $n \geq \log_2 d$. The assemblage between Charlie and Dani with elements $\{\sigma_{d|w}^{\mathsf{R}}\}$ that we'll aim to self-test will be shown to belong to a Hilbert space of dimension $2^n$.

\bigskip

The second step follows from  noticing that the $n$-fold tensor products of Pauli matrices and the qubit identity matrix form a basis for the space of Hermitian operators in $\cH_{2^n}$. For example, for $n=3$, the such a basis is $\{O_1 \otimes O_2 \otimes O_3 \, \vert \, O_1,O_2,O_3 \in \{\id_2,X,Y,Z\} \}$. From the qubit case discussed in the main text one learns that what we need to self-test in $\cH_{C'}$ is an assemblage that includes the elements of a basis for the space of Hermitian operators in $\cH_{d}$. After the embedding has been performed, then, one can achieve this by self-testing tensor products of Pauli operators. Hence, this second step pertains to framing the high-dimensional self-testing question as a protocol to self-test tensor product of Paulis (or more precisely, the tensor product of their eigenvectors). 

Once the elements of the assemblage $\{\sigma_{d|w}^{\mathsf{R}}\}$ are self-tested to be the tensor product of Pauli eigenvectors, one may express the Hermitian operators $\mathcal{F}_{abxy}$ from a generic steering functional as $\mathcal{F}_{abxy} = \sum_{d,w} f_{abdxyw} \, \widetilde{\sigma}_{d|w}^{\mathsf{R}}$\,, where the assemblage elements $\widetilde{\sigma}_{d|w}^{\mathsf{R}}$ are indeed the products of Pauli eigenvalues without the contributions that come from their transpose (see next point). Similarly to the qubit case, the Bell-type functional to evaluate reads
\begin{align}\label{pseudobell2}
    \mathcal{I}[\textbf{p}] = \sum_{abdxyw} f_{abdxyw} \, p(ab0d|xy\star w)\,.
\end{align}

\bigskip

The third step pertains to the issue of self-testing discussed in the main text: one cannot self-test a state, but rather a convex combination of a state and its transpose. This was not an issue for the qubit case, since it didn't give rise to false-positive outcomes in the protocol (i.e., a negative value of Eq.~\eqref{pseudobell} certifies post-quantumness of the assemblage shared between Alice, Bob, and Charlie for any value of the unknown parameter $r$ in Eq.~\eqref{eq:theSTdass}). However, for the case of qudits, this needs more consideration. Let us illustrate it for the case of $n=2$. Here, by performing a parallel self-testing of qubit states (that is, two self-testing protocols in parallel of one qubit each), one can conclude that the state being self-tested has the form $(r \rho_1 + (1-r) \rho_1^\mathsf{T}) \otimes (r' \rho_2 + (1-r') \rho_2^\mathsf{T})$, where $\rho_1 \otimes \rho_2$ is the target state we aimed to self-test, and $r,r' \in [0,1]$ are unknown parameters. Now, self-testing a state that corresponds to a convex combination of $\rho_1 \otimes \rho_2$ and $\rho_1^\mathsf{T} \otimes \rho_2^\mathsf{T}$ will not create false-positives, i.e., will not enable quantum assemblages to produce negative values in the Bell-type functional, just like was argued for the case of qubits in the main text. However, should the state of the system indeed contain terms corresponding to $\rho_1 \otimes \rho_2^\mathsf{T}$ and $\rho_1^\mathsf{T} \otimes \rho_2$, then one wouldn't be able to conclude post-quantum steering when observing a negative value of the Bell functional. Hence, one needs to take extra care to ``align'' the reference frames for each self-test, so that when one test self-tests the target state $\rho_1$ then the other protocols for all the other qubits also self-test the target state, and when one self-test $\rho_1^\mathsf{T}$ for the first qubit, then all the other self-testing protocols will also self-test the transpose of their corresponding target states. To summarise, we need to complement the self-testing protocol for single-qubits described in the main text with something extra so that the state that we successfully self-test (for each assemblage element) has the form $r (\rho_1 \otimes \cdots \otimes \rho_n) + (1-r) (\rho_1 \otimes \cdots \otimes \rho_n)^\mathsf{T}$. A detailed protocol to perform such `synchronised' self-testing was presented in Ref.~\cite{bowles2018self}, and can be straightforwardly applied here. When doing so, then, the input and output sets of Charlie and Dani in the scenario of Fig.~\ref{fig:bowlset} have a specific form. In the case of Charlie, his measurements are labelled by
\begin{align*}
    z \in \{1,2,3,4,5,6\}^{\times n} \cup \{\star, \lozenge, \blacklozenge\} =: \mathbb{Z}\,,
\end{align*}
where intuitively the values $\{1,2,3,4,5,6\}^{\times n}$ are use for self-testing,  the values $\{\lozenge, \blacklozenge\}$ are used to `align the reference frames', and the value $\star$ is used to evaluate the Bell-type functional in Eq.~\eqref{pseudobell2}. The set of outputs for Charlie's measurements have a specific form that is not worth explicitly mentioning here, beyond the fact that measurement $z=\star$ has only two outcomes $c \in \{0,1\}$. In the case of Dani, it is useful to label his measurement choices as the array $w \in \{1,2,3\}^{\times n}$ and his outputs as $d \in \{0,1\}^{\times n}$.

\bigskip

Finally, once the self-testing step from Ref.~\cite{bowles2018self} has been successful, the linear functional of Eq.~\eqref{pseudobell2} may be expressed as: 
\begin{align}\label{eq:final}
    \mathcal{I}[p] = \tfrac{r}{2^n} \sum_{abdxyw} f_{abdxyw} \,  &\Tr{\left(\bigotimes_{i=1}^n \widetilde{\sigma}^{\mathsf{R}}_{d_i|w_i}\right)^{\mathsf{T}}  \sigma_{ab|xy} }{\cH_C} \\ &+ \tfrac{1-r}{2^n} \sum_{abdxyw} f_{abdxyw} \,  \Tr{\bigotimes_{i=1}^n \widetilde{\sigma}^{\mathsf{R}}_{d_i|w_i} \,  \sigma_{ab|xy} }{\cH_C}\,, \nonumber
\end{align}
where we have taken Charlie's measurement operator for $z=\star$ and $c=0$ to be $M_{0} = \ket{\Psi_n}\bra{\Psi_n}$ with $\ket{\Psi_n}=\frac{1}{\sqrt{2^n}} \sum_{k=0}^{2^n-1} \ket{kk}_{CC'}$. A similar reasoning to that for the qubit case in the main text follows: (i) the first term in Eq.~\eqref{eq:final} will be positive by definition for all quantum assemblages $\{\sigma_{ab|xy}\}$, (ii) the second term in Eq.~\eqref{eq:final} will be positive for all quantum assemblages $\{\sigma_{ab|xy}\}$ since $\sum_{abdxyw} f_{abdxyw} \,  \Tr{\left(\bigotimes_{i=1}^n \widetilde{\sigma}^{\mathsf{R}}_{d_i|w_i}\right)^{\mathsf{T}}  \,  \left(\sigma_{ab|xy}\right)^{\mathsf{T}}  }{\cH_C}$ is positive for any quantum assemblage with elements $\{\left(\sigma_{ab|xy}\right)^{\mathsf{T}} \}$, and the assemblage $\{\left(\sigma_{ab|xy}\right)^{\mathsf{T}} \}$ is quantum if the assemblage  $\{\sigma_{ab|xy}\}$ is quantum \cite[Sec.~5.1]{sainz2018formalism}. 

One can next get rid of the assumption that $M_{0} = \ket{\Psi_n}\bra{\Psi_n}$ following the same techniques presented in the Appendix for the case of qubits. It then follows that, following a successful self-test, a negative value of the functional of Eq.~\eqref{pseudobell2} certifies that the assemblage shared between Alice, Bob, and Charlie has no quantum realisation.

\section{Derivation of the Independence Assumption for $r \in \{0,1\}$}\label{ap:indasump}

Here we discuss how the form $\mathbb{\Sigma}_{AB}\otimes \mathbb{\Sigma_{D}}$
for the assemblage $\mathbb{\Sigma}_{ABC}$ can be derived for the cases where $r \in \{0,1\}$ in the self-testing argument in Eq.~\eqref{eq:theSTdass}. First we need to establish that when $r \in \{0,1\}$ in Eq.~\eqref{eq:theSTdass}, the assemblage is extremal, i.e. it cannot be written as a convex combination of other assemblages. We will give the argument for $r=0$, since the argument for $r=1$ is  the same up to obvious changes. For clarity, when $r=0$ we relabel the assemblage elements in Eq.~\eqref{eq:theSTdass} as $\sigma_{d|w}^{*}:=\tilde{\sigma}^{\mathsf{R}}_{d|w}\otimes|0\rangle\langle 0|$. 

To show that the assemblage with elements $\sigma_{d|w}^{*}$ is extremal, we can construct a steering expression, which when equal to $0$, intersects with a single assemblage, ensuring extremality. In other words, since the set of bipartite non-signalling assemblages (equivalently quantum assemblages) is convex and bounded, if a hyperplane intersects this convex body at one point, then this point is extremal in that convex body. The steering expression constructed from assemblage elements $\sigma^*_{d|w}$ is
\begin{equation}
\sum_{d,w}\textrm{tr}\left(F_{dw}\sigma_{d|w}\right)
\end{equation}
where $F_{dw}=2\sigma^*_{d\oplus 1|w}+\mathbb{I}\otimes|1\rangle\langle 1|$. It remains to show that when this expression is equal to $0$, then $\sigma_{d|w}=\sigma^*_{d|w}$ for all elements $\sigma_{d|w}$. 

First note that the expression must be non-negative for all non-signalling assemblages, because $F_{dw}$ is a rank one projector, $\textrm{tr}\left(F_{dw}\sigma_{d|w}\right)\geq 0$, since it is a probability. Thus, a sum of non-negative terms is non-negative, and the only way the sum can be zero is if every term is equal to zero. If $\textrm{tr}\left(F_{dw}\sigma_{d|w}\right)=0$, then $\sigma_{d|w}$ is in the kernel of the projector $F_{dw}=2\sigma^*_{d\oplus 1|w}+\mathbb{I}\otimes|1\rangle\langle 1|$, thus $\sigma_{d|w}=\alpha\sigma^{*}_{d|w}$ for a non-negative real number $\alpha$. Note that $\sigma_{d|w}+\sigma_{d\oplus 1|w}=\alpha\sigma^{*}_{d|w}+\alpha'\sigma^{*}_{d\oplus 1|w}=\rho_{C}$, which is the density operator of Charlie's system. Thus, $\alpha+\alpha'=2$, and for the three different inputs $w$, we have $\sigma_{d|w}+\sigma_{d\oplus 1|w}=\sigma_{d|w'}+\sigma_{d\oplus 1|w'}=\rho_{C}$. The only density operator $\rho_{C}$ that will satisfy these constraints is $\frac{1}{2}\mathbb{I}\otimes|0\rangle\langle0|$, thus implying $\alpha=1$ for all assemblage elements $\sigma_{d|w}=\alpha\sigma^{*}_{d|w}$. We can conclude then that the assemblage with elements $\sigma^{*}_{d|w}$ is extremal. Analogously we can show that for $r = 1$ in Eq.~\eqref{eq:theSTdass} is also extremal.

Equipped with the fact that the assemblages are extremal, we can prove that Charlie's two systems are independent in the following theorem.

\begin{thm} 
 Consider an assemblage with elements $\{\sigma_{abd|xyw}\}$ prepared in Charlie's Hilbert space $\cH_{C}\otimes\cH_{C'}$ after outputs produced by Alice, Bob, and Dani. If the following are satisfied by this assemblage:

\begin{compactitem}
    \item[(i)] $\forall x, y$, $\sigma_{d|w}=\textrm{tr}_{C'}\left\{\sum_{ab}\sigma_{abd|xyw}\right\}$,
    \item[(ii)] $\sum_{ab}\sigma_{abd|xyw}= \psi\otimes \sigma_{d|w}$, where $\{\sigma_{d|w}\}$ are elements of an extremal assemblage and $\psi$ corresponds to a pure state in $\cH_{C'}$.
\end{compactitem}

For such an assemblage, it then follows that 
\begin{align}
\sigma_{abd|xyw}=\sigma_{ab|xy}\otimes\sigma_{d|w}
\end{align}
\end{thm}

\begin{proof}
First we prove that $\sigma_{abd|xyw}=\sigma'_{abd|xyw}\otimes \sigma_{d|w}$ for some appropriate assemblage elements $\{\sigma'_{abd|xyw}\}$. To do so we introduce notation 
\begin{align*}
\sigma_{d|abxyw}:=& \frac{\sigma_{abd|xyw}}{p(ab|xyw)}=\frac{\sigma_{abd|xyw}}{p(ab|xy)},
\end{align*}
where the second equality holds by the no-signalling condition. Using this notation it follows that:
\begin{align*}
\sigma_{d|w}=&\textrm{tr}_{C'}\left\{\sum_{ab}\sigma_{abd|xyw}\right\}
=\textrm{tr}_{C'}\left\{\sum_{ab}p(ab|xy)\sigma_{d|abxyw}\right\} \\
&=\textrm{tr}_{C'}\left\{\sum_{xy}\frac{1}{XY}\sum_{ab}\frac{p(abxy)}{p(xy)}\sigma_{d|abxyw}\right\}
=\sum_{\lambda}p(\lambda)\textrm{tr}_{C'}\left\{\sigma_{d|w\lambda}\right\}\,, 
\end{align*}
where we have relabeled the variables as $\lambda := abxy$, and for simplicity that $p(xy)=\frac{1}{XY}$.

By assumption, the assemblage with elements $\{\sigma_{d|w}\}$ is extremal, and hence it must happen that 
\begin{align*}
    \sigma_{d|w} = \textrm{tr}_{C'}\left\{\sigma_{d|w\lambda}\right\} \quad \forall \lambda\,.
\end{align*}
Therefore, 
\begin{align*}
    \sigma_{d|w} = \frac{1}{p(ab|xy)} \textrm{tr}_{C'}\left\{\sigma_{abd|xyw}\right\} \quad \forall \, a,b,x,y\,.
\end{align*}
and thus
\begin{align*}
\sigma_{abd|xyw}=\sigma_{abd|xyw}'\otimes \sigma_{d|w}\,,
\end{align*}
where the elements $\{\sigma_{d|w}\}$ belong to operators acting on $\cH_{C}$, and the elements $\{\sigma_{abd|xyw}'\}$ belong to operators on $\cH_{C'}$. 

 Now the final step is to show that $\sigma_{abd|xyw}' = \sigma_{ab|xy}$.
For this, notice that 
\begin{align*}
    \sum_{ab} \sigma_{abd|xyw} =\left( \sum_{ab} \sigma_{abd|xyw}' \right)\otimes \sigma_{d|w} = \psi \otimes \sigma_{d|w} \,,
\end{align*}
where the last equality follows from property (ii) in the statement of the Theorem. Therefore, $\sum_{ab} \sigma_{abd|xyw}'$ must be independent of $d$ and $w$, which implies that $\sigma_{abd|xyw}'$ is independent from $d$ and $w$. Since, moreover, $\sigma_{ab|xw}=\textrm{tr}_{C}\left\{\sum_{d}\sigma_{abd|xyw}\right\}$, it follows that $\sigma_{abd|xyw}' = \sigma_{ab|xy}$.

\end{proof}
\end{document}